\documentclass[journal]{IEEEtran}

\usepackage[colorlinks]{hyperref}

\usepackage{setspace}
\usepackage{mathtools,enumitem,color,bbm}
\usepackage{graphicx,tabularx,array}
\usepackage{hyperref}

\usepackage{amsthm,mathrsfs}
\usepackage{makecell}
\usepackage{multirow}
\usepackage[T1]{fontenc}
\usepackage{cite}
\usepackage{txfonts}
\usepackage{graphicx}
 \usepackage{subcaption}
 
\IEEEoverridecommandlockouts

\definecolor{green}{rgb}{0.6627,0.8196,0.5568}
\definecolor{purple}{rgb}{ 0.7647,    0.6078,    0.8824}

\newcommand{\discardpages}[1]{
  \xdef\discard@pages{#1}
  \AtBeginShipout{
    \renewcommand*{\do}[1]{
      \ifnum\value{page}=##1\relax%
        \AtBeginShipoutDiscard
        \gdef\do####1{}
      \fi%
    }%
    \expandafter\docsvlist\expandafter{\discard@pages}
  }%
}
\newif\ifkeeppage
\newcommand{\keeppages}[1]{
  \xdef\keep@pages{#1}
  \AtBeginShipout{
    \keeppagefalse%
    \renewcommand*{\do}[1]{
      \ifnum\value{page}=##1\relax%
        \keeppagetrue
        \gdef\do####1{}
      \fi%
    }%
    \expandafter\docsvlist\expandafter{\keep@pages}
    \ifkeeppage\else\AtBeginShipoutDiscard\fi
  }%
}

\DeclarePairedDelimiter{\ceil}{\lceil}{\rceil}
\DeclarePairedDelimiter\floor{\lfloor}{\rfloor}

\newtheorem{Theorem}{Theorem}
\newtheorem{Proposition}{Proposition}

\newtheorem{Example}{Example}
\newtheorem{Definition}{Definition}

\ifCLASSINFOpdf

\begin{document}
%


\title{Non-Linear Analog Processing in MIMO Systems with
Coarse Quantization}

\author{%
  \IEEEauthorblockN{Marian Temprana Alonso, \textit{Graduate Student Member, IEEE},
                    Xuyang Liu, \textit{Graduate Student Member, IEEE},
                    Hamidreza Aghasi, \textit{member, IEEE},
                    Farhad Shirani, \textit{member, IEEE}
                   }\thanks{This work was presented in part at IEEE
      International Symposium on Information Theory and IEEE Global Communications Conference \cite{Shirani2022,shirani2022quantifying}. This work was supported by NSF grants CCF-2241057 and CCSS-2242700/2233783.
      \\ (Corresponding author: M. Temprana Alonso). 
      \\ M. Temprana Alonso is with School of Computing  and Information Sciences, Florida International University, Miami, FL (email: mtemp009@fiu.edu)
      \\ X. Liu is with Department of Electrical Engineering and Computer Science,  University of California, Irvine (email: xuyanl3@uci.edu)
      \\ H. Aghasi is with Department of Electrical Engineering and Computer Science,  University of California, Irvine (email: haghasi@uci.edu)
      \\ F. Shirani is with School of Computing  and Information Sciences, Florida International University, Miami, FL (email: fshirani@fiu.edu)}}



%


\maketitle

 \begin{abstract}
Analog to digital converters (ADCs) are a major contributor to the power consumption of multiple-input multiple-output (MIMO) receivers in large bandwidth millimeter-wave systems. \textcolor{black}{Prior works have considered two mitigating solutions to reduce the ADC power consumption: i) decreasing the number of ADCs via analog and hybrid beamforming, and ii) decreasing the ADC resolution, i.e., utilizing one-bit and few-bit ADCs.} 
 These mitigating solutions lead to performance loss in terms of achievable rates due to increased quantization error. In this work, the use of nonlinear analog operators such as envelope detectors and polynomial operators, prior to sampling and quantization is considered, as a way to reduce the aforementioned rate-loss. The receiver architecture consists of linear combiners, nonlinear analog operators, and few-bit ADCs.
The fundamental performance limits of the resulting communication system, in terms of achievable rates,  are investigated under various assumptions on the set of implementable analog operators. \textcolor{black}{Extensive numerical evaluations are provided to evaluate the set of achievable rates and the power consumption of the proposed receiver architectures.} Circuit simulations and measurement results, based on both 22 nm FDSOI CMOS technology and 65 nm Bulk CMOS transistor technologies, are provided to justify the power efficiency of the proposed receiver architectures. 
\end{abstract}


%
\IEEEpeerreviewmaketitle

\section{Introduction}
In order to accommodate the demand for higher data-rates, the wireless spectrum has been continuously expanding over the past several decades. The millimeter-wave (mm-wave) spectrum is utilized in
the fifth generation (5G) wireless systems to allow for larger channel bandwidths
compared to earlier generation radio frequency systems, which operate in frequencies {below 6 GHz \cite{6GHz}}.
The energy consumption of components such as analog to digital converters (ADCs) increases significantly in mm-wave systems. In theory, the power consumption of an ADC grows linearly with bandwidth, and the rate of increase is even more significant in practice due to the excessive loss associated with the passive components at higher frequencies, which causes an abrupt drop in ADC energy-efficiency as the bandwidth is pushed past 100 MHz  \cite{ADCpower}.
 For instance,
the power consumption of commercial
high-speed ($\geq$ 20 GSample/s), high-resolution
(8-12 bits) ADCs is around 500 mW per ADC \cite{zhang2018low}.  

\textcolor{black}{The ADC power consumption is particularly prohibitive in mm-wave systems utilizing large antenna arrays and digital beamforming archiectures.} To elaborate, in order to mitigate the inherent high isotropic path loss at high frequencies, mm-wave systems utilize {directive} narrow-beams generated by large antenna arrays \cite{rappaport2015millimeter}. For instance, 5G wireless networks  envision hundreds of antennas at the base station (BS) and in excess of ten antennas at the user-end (UE) \cite{akram2021massive}. 
In conventional multiple-input multiple-output (MIMO) systems with digital beamforming, each antenna output is digitized separately by a dedicated ADC. This requires a large number of ADCs which are a significant source of power consumption in large bandwidth MIMO receivers \cite{heath2016overview,mendez2015channel}. \textcolor{black}{Analog and hybrid beamforming have been proposed to mitigate the ADC power consumption by reducing the number of ADCs. 
In hybrid beamforming, the receiver uses a set of analog beamformers in the radio frequency (RF) domain to combine the large number of analog signals at the receiver antennas and feed them to a small set of ADCs \cite{molisch2017hybrid,heath2016overview,bernardo2022}. 
In analog beamforming, a single beamformer is used to linearly combine the received signals \cite{venkateswaran2010analog,ning2021prospective}.}

\textcolor{black}{In addition to these beamforming architectures, an alternative approach to reduce ADC power consumption is to reduce the resolution of the ADCs. To elaborate, in the standard ADC design, power consumption is proportional to the number of quantization bins and hence grows exponentially in the number of output bits \cite{walden1999analog}, which prohibits the use of high resolution ADCs. There has been extensive recent efforts to design receiver architectures and coding strategies using 
 analog, hybrid, and digital beamforming with few-bit ADCs  \cite{li2018achievable,ning2021prospective,han2021hybrid,molisch2017hybrid,heath2016overview,khalili2020throughput,dutta2020capacity,jacobsson2017throughput,mo2015capacity,mo2017channel}.  The restriction to few-bit ADCs leads to reduced communication rates compared to when high resolution ADCs are utilized.
 In hybrid beamforming with few-bit ADCs, simple analog processing, linear processing in particular, is used to partially mitigate the rate loss due to low resolution quantization \cite{zhao2020energy,roth2018comparison,mo2015capacity,khalili2018mimo,yu2023low}.  }
 In \cite{zirtiloglu2022}, a power-efficient hybrid MIMO receiver is presented where the analog and digital processing are jointly optimized by using task-specific quantization techniques. The hybrid beamforming approach was further extended in  \cite{khalili2021mimo}, where hybrid blockwise architectures, consisting of delay elements, were considered. This allows for temporal linear processing of signals in the analog domain, providing additional degrees of freedom in choosing the processing function. 
Adaptive thresholds ADCs were considered in  \cite{khalili2021mimo}, which allow for modifying the ADC thresholds based on past quantization outputs. This  method has similarities with successive approximation register (SAR) ADCs \cite{suarez1975SAR}.  The latter two approaches improve the channel capacity compared with conventional hybrid beamforming architectures, however, the practical implementation of high precision analog delay elements and SAR ADCs with high sampling rates is challenging due to synchronization issues and delay accuracy limitations.

 %
 %

In this work, we argue that digital, hybrid, and analog beamforming approaches suffer from a phenomenon which we call the \text{curse of low dimensions} (Section \ref{sec:curse_of_low}). We quantify the rate-loss due to this phenomenon, and provide solutions to mitigate it. Particularly, we consider nonlinear analog operators --- such as envelope detectors and low degree polynomial operators --- prior to sampling and quantization, as a way to mitigate the rate-loss due to coarse quantization. The power consumption and circuit design for {receiver architectures deploying} these nonlinear operators based on measurements and simulations in 22 nm and 65 nm CMOS technologies are provided in Section \ref{sec:cir}.

The main contributions of this work are summarized below:
\vspace{-.15in}
\color{black}
\begin{itemize}[leftmargin=*]
   \item To characterize the channel capacity under analog beamforming when envelope detectors are used for analog signal processing. 
(Theorems \ref{th:1}-\ref{th:4}). 
  \item To characterize the capacity under analog beamforming when polynomial operators are used. 
(Theorem \ref{th:5}). 
\item To introduce a receiver architecture for hybrid beamforming using envelope detectors, and to provide the high signal-to-noise ratio (SNR) capacity and inner-bounds to the low SNR achievable rates in terms of the total ADC power budget $P_{ADC}$, number of output levels of each ADC, $\ell$, and the polynomial degree, $\delta_{poly}$ (Theorem \ref{th:6}).
    \item To provide circuit designs and associated performance simulations for implementing polynomials of degree up to four and pairs of concatenated envelope detectors, and to evaluate their power consumption. (Sections \ref{sec:num} and  \ref{sec:cir})
    \item To evaluate the rate-power tradeoff under linear and nonlinear analog processing using the rate derivations in our theoretical analysis and the power figures in our circuit simulations.  (Sections \ref{sec:num} and  \ref{sec:cir})
\end{itemize}
\color{black}

\color{black}
{\em Notation:}
Sets are denoted by calligraphic letters such as $\mathcal{X}$, families of sets by sans-serif letters such as $\mathsf{X}$. 
The set $\{1,2,\cdots, n\}$ is represented by $[n]$. $\mathcal{X}^c$ denotes the complement of $\mathcal{X}$.
The vector $(x_1,x_2,\hdots, x_n)$ is written as $x^n$, and $(x_k,x_{k+1},\cdots,x_n)$ is denoted by $x_{k}^n$. The $i$th element of $x^n$ is written as $x_i$.  An $n\times m$ matrix is written as $h^{n\times m}=[h_{i,j}]_{i\in [n], j\in [m]}$, its $i$th row is denoted by $h_{i,:}=[h_{i,j}]_{j\in [m]}$ and its $j$th column is $h_{:,j}=[h_{i,j}]_{i\in [n]}$. The $n\times n$ identity matrix is denoted by $\mathbf{I}_n$.
We use bold-face letters such as $\mathbf{x}$ and $\mathbf{h}$ instead of $x^n$ and $h^{n\times m}$, respectively, to represent vectors and matrices when the dimension is clear from the context.  $\mathbf{x}^H$ denotes the hermitian of $\mathbf{x}$. We write $||\cdot||_2$ to denote the $L_2$-norm.  
 Upper-case letters such as $X$ represent random variables, and lower-case letters such as $x$ represent their realizations.  Similarly, random vectors and random matrices are denoted by upper-case letters such as $\mathbf{X}$  and $\mathbf{H}$, respectively. For a Gaussian random vector $\mathbf{X}$ 
with mean vector $\pmb{\mu}$ and covariance matrix $\mathbf{\Sigma}$, we write $\mathbf{X}\sim\mathcal{N}( \pmb{\mu},\mathbf{\Sigma})$.

\section{Problem Formulation}
\label{sec:mot}
  \begin{figure}[t]
\centering 
\includegraphics[width=0.4\textwidth]{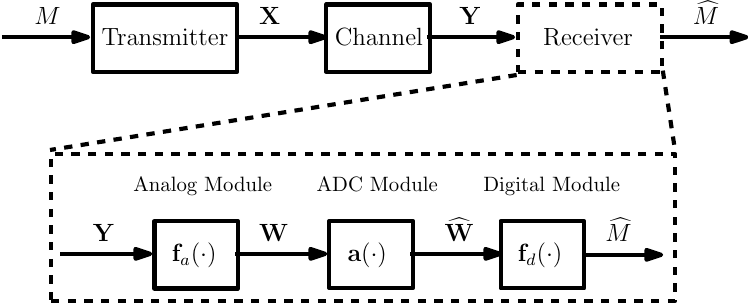}
\caption{
\textcolor{black}{(Top) MIMO communication setup consisting of message $M$, channel input $\mathbf{X}$, channel output $\mathbf{Y}$, and message reconstruction $\widehat{M}$. (Bottom) receiver consisting of an analog processing module $\mathbf{f}_a(\cdot)$, ADC input $\mathbf{W}$,
ADC module $\mathbf{a}(\cdot)$, ADC output $\mathbf{\widehat{W}}$, and a digital processing module $\mathbf{f}_{d}(\cdot)$.}}
\vspace{-.2in}
\label{fig:classic}
\end{figure}
\label{sec:form}
\subsection{System Model}
We consider a MIMO channel with $n_t$ transmit antennas and $n_r$ receive antennas. The input and output\footnote{To simplify notation, we have considered real-valued variables. The derivations can be extended to complex variables in a straightforward manner.} 
$(\mathbf{X}, \mathbf{Y})\in \mathbb{R}^{n_t}\times \mathbb{R}^{n_r}$ are related through $
\mathbf{Y}=\mathbf{h}\mathbf{X}+\mathbf{N}$, where $\mathbf{h}\in \mathbb{R}^{n_r\times n_t}$ is the 
(fixed) channel gain matrix
and 
$\mathbf{N}\sim \mathcal{N}(0,I_{n_r})$ is a vector of Gaussian variables with zero mean and unit variance. The channel input has average power constraint $P_T$, i.e., $\mathbb{E}(\|\mathbf{X}\|_2^2)\leq P_T$. 

The receiver model is shown in the bottom figure in Figure \ref{fig:classic}. In general, upon receiving the channel output $\mathbf{Y}$, the receiver may perform analog processing operations. For instance, in hybrid beamforming the receiver applies an affine transformation on the received signal using an analog linear combiner, and in analog beamforming, phase shifters are used to produce a linear combination of the received signals. 
The analog processing operations are represented by the function $\mathbf{f}_{a}(\cdot)$. The output of $\mathbf{f}_a(\cdot)$ has dimension equal to the number of ADCs available at the receiver. That is, the function $\mathbf{f}_a(\cdot)$ takes the $n_r$-dimensional vector $\mathbf{Y}$ as input, and outputs an $n_q$-dimensional vector $\mathbf{W}=\mathbf{f}_a(\mathbf{Y})$, where $n_q$ is the number of ADCs. 
We denote $\mathbf{f}_a= (f_1,f_2,\cdots,f_{n_q})$, where $f_j:\mathbb{R}^{n_r}\to \mathbb{R}, j\in [n_q]$. 
The analog vector $\mathbf{W}$ is fed to the ADCs to produce the digitized signal $\mathbf{\widehat{W}}$.
The operation of the ADCs is described in Section \ref{sec:ADC}.
The receiver then performs blockwise digital processing on the ADC outputs to produce the reconstruction $\widehat{M}$ of the transmitted message $M$. The digital processing operations performed to decode the message are represented by $\mathbf{f}_d(\cdot)$.


\subsection{The ADC Module}
\label{sec:ADC}
We assume that the receiver is equipped with a collection of ADCs. To simplify the notation, we assume that each of the ADCs have $\ell$ output levels, i.e., the ADCs have the same resolution.  An ADC is modeled as a mapping from  a continuous-valued input $y$ to a discrete-valued output $\hat{y}$. For an ADC with $\ell$ output levels, we define a threshold vector $t^{\ell-1}=(t_1,t_2,\cdots,t_{\ell-1})$. The output 
is then given by:
 \begin{align}
     \hat{y}= 
     \begin{cases}
         0 \qquad &\text{if } \quad  y<t_1\\
         i &\text{if } \quad \exists i\in [\ell-2], t_i\leq y < t_{i+1}\\
         \ell-1 &\text{if } \quad t_{\ell-1}\leq y 
     \end{cases},
     \label{eq:ADC}
 \end{align}
 where $t_1<t_2<\cdots<t_{\ell-1}$.
 For a receiver equipped with $n_q$ ADCs each with $\ell$ output levels, the threshold matrix is defined as $t^{n_q\times {(\ell-1)}}$, where $t_{i,:}, i\in [n_q]$ is the threshold vector corresponding to the $i$th ADC. 

Theoretically, the power consumption of an ADC is equal to the product of its capacitive load C, voltage V, sampling frequency $\omega_s$, and number of output levels $\ell$ \cite{BR}. That is, the power consumption of a single $\ell$-level ADC is $\alpha\ell$, where we have defined $\alpha\triangleq CV\omega_s$. 
Given a total ADC power budget $P_{ADC}$ and number of output levels $\ell$, the maximum number of ADCs at the receiver is  given by $n_q=\floor{\frac{P_{ADC}}{\alpha \ell}}$.

\subsection{Conventional Beamforming Architectures}
\label{sec:arch}
\noindent \textbf{Analog Beamforming:} Analog beamforming utilizes analog phase shifters and only one RF chain 
for the beamforming operation. This leads to a simplified design and low power consumption compared to hybrid and digital beamforming. However,  it potentially yields lower rates. The analog processing function in analog beamforming is represented as:
\begin{align*}
    W=\mathbf{f}_a(\mathbf{Y})= \mathbf{b}^{T}\mathbf{Y},
\end{align*}
where $\mathbf{b}= (b_1,b_2,\cdots,b_{n_r})$ is the analog beamforming vector\footnote{By default, we assume column vectors unless indicated otherwise.} and $\mathbf{b}^T$ denotes the transpose of $\mathbf{b}$. 

\noindent\textbf{Hybrid Beamforming:} Hybrid beamforming uses a collection of analog 
beamformers  to linearly combine the received signals in the analog domain \cite{zhang2005variable,alkhateeb2014mimo}.
This linear analog signal processing improves performance, in terms of achievable rates, by rotating the received signal such that the information-loss in the quantization step is reduced 
\cite{koch2013low,khalili2018mimo}.    The analog processing function in hybrid beamforming is represented as:
\begin{align}
    \mathbf{W}=\mathbf{f}_a(\mathbf{Y})= \mathbf{v}\mathbf{Y},
    \label{eq:hyb0}
\end{align}
where $\mathbf{v}= [v_{i,j}]_{i\in [n_q], j\in [n_r]}$ is the hybrid beamforming matrix.

\noindent\textbf{Digital Beamforming:} In digital beamforming, each antenna is directly connected to its dedicated ADC, i.e., $n_r=n_q$ and $\mathbf{f}_a(\cdot)$ is the identity function \cite{mo2015capacity,MIMO2,nossek2006capacity}. This leads to an increased number of ADCs and potentially higher achievable rates at the expense of increased power consumption.


\section{Curse of Low Dimensions}
\label{sec:curse_of_low}
\begin{figure*}[t]
\centering
  \begin{subfigure}[b]{0.18\textwidth}
  \includegraphics[width=\linewidth]{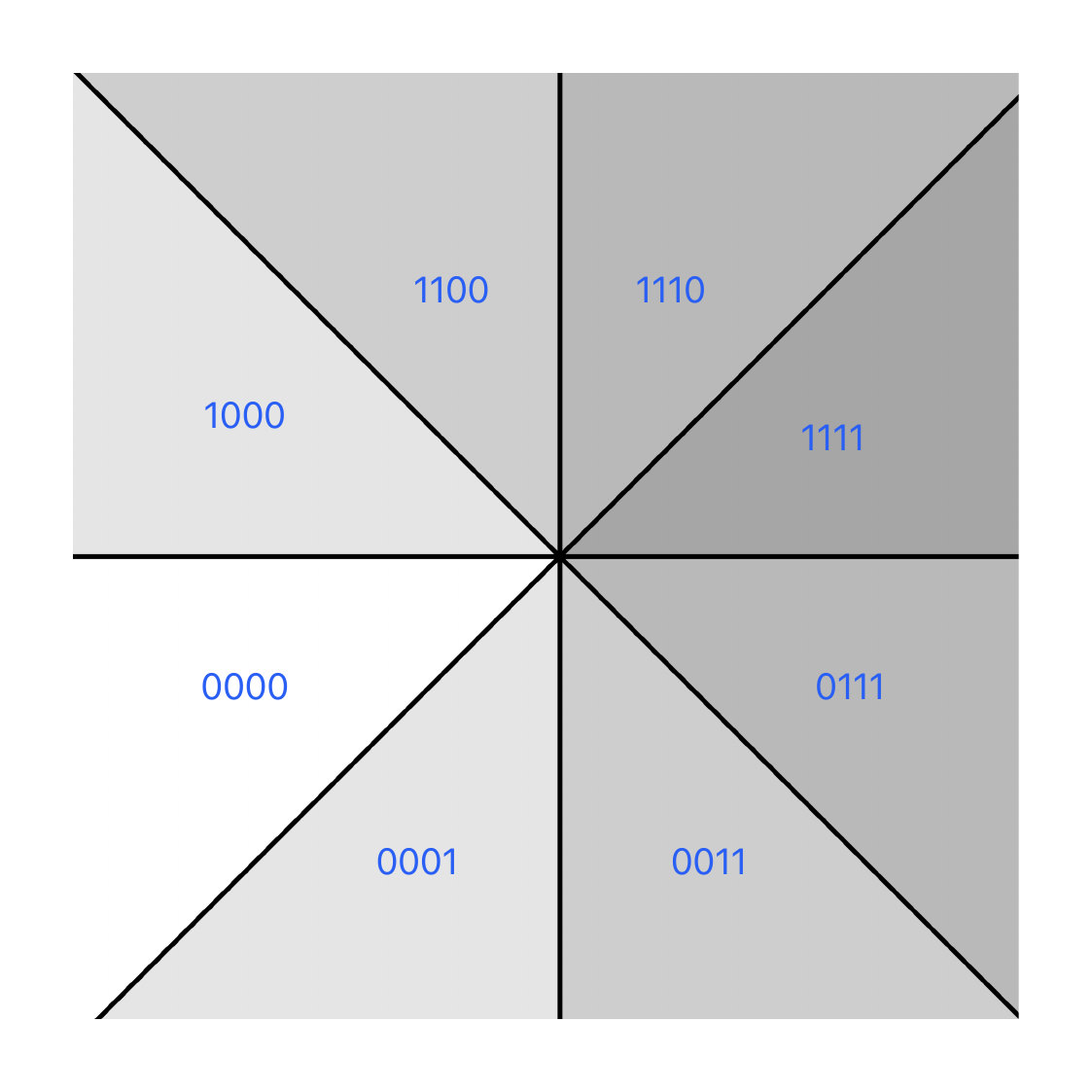}
    \caption{PSK}
  \end{subfigure}
  \hfill
  \begin{subfigure}[b]{0.18\textwidth}
    \includegraphics[width=\linewidth]{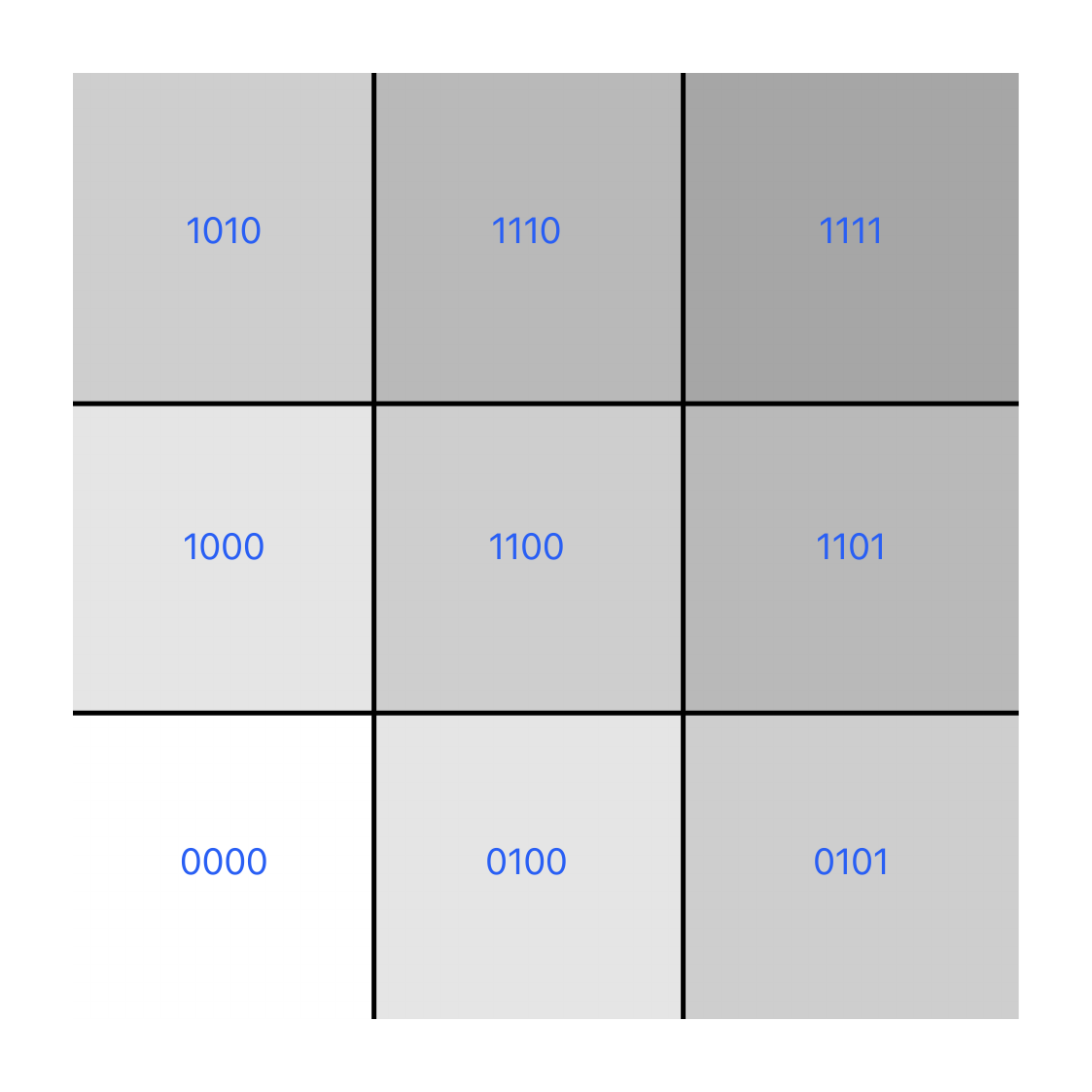}
    \caption{QAM}
  \end{subfigure}
  \hfill
  \begin{subfigure}[b]{0.18\textwidth}
    \includegraphics[width=\linewidth]{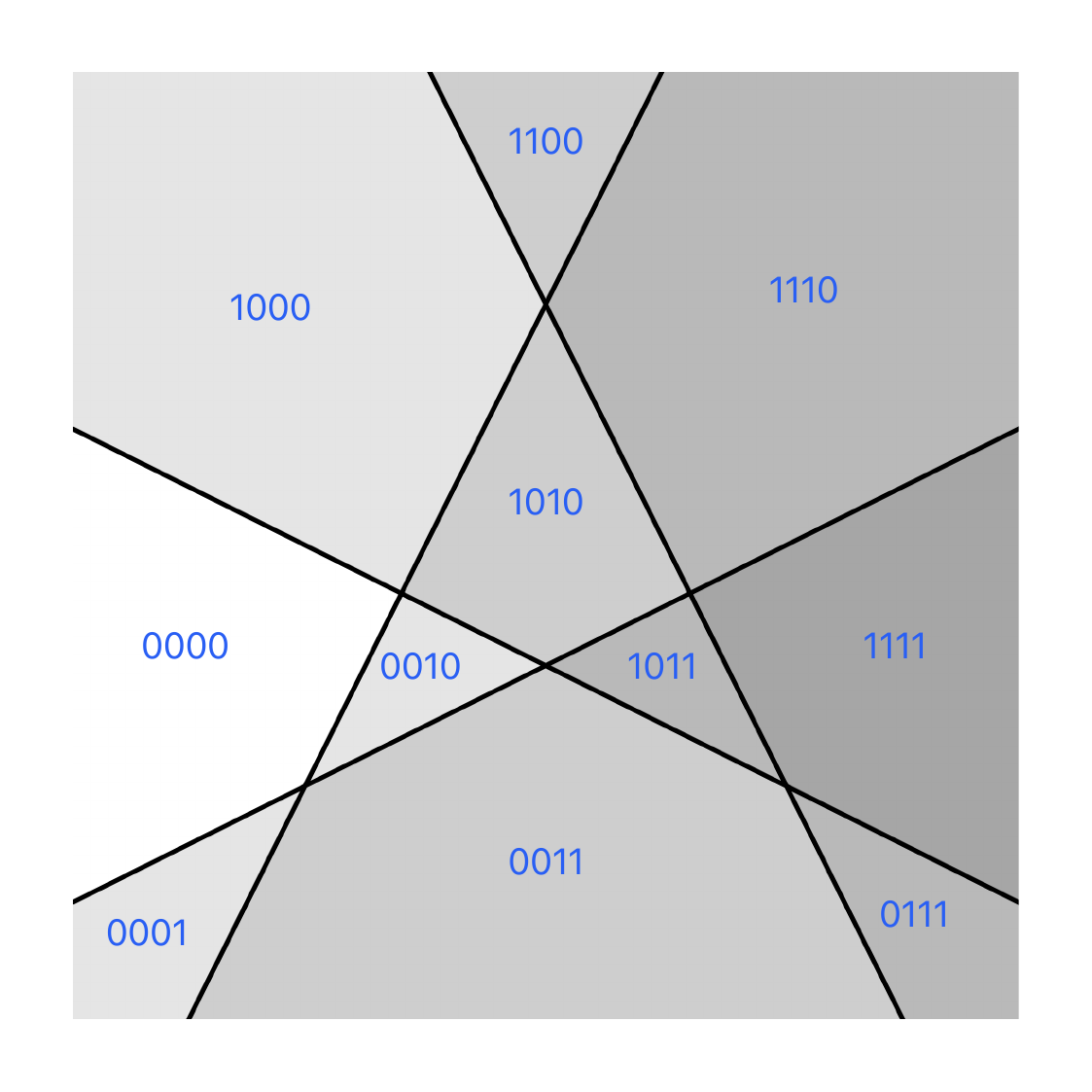}
    \caption{General Position \cite{khalili2021mimo}}
  \end{subfigure}
  \hfill
    \begin{subfigure}[b]{0.18\linewidth}
   { \includegraphics[width=\linewidth]{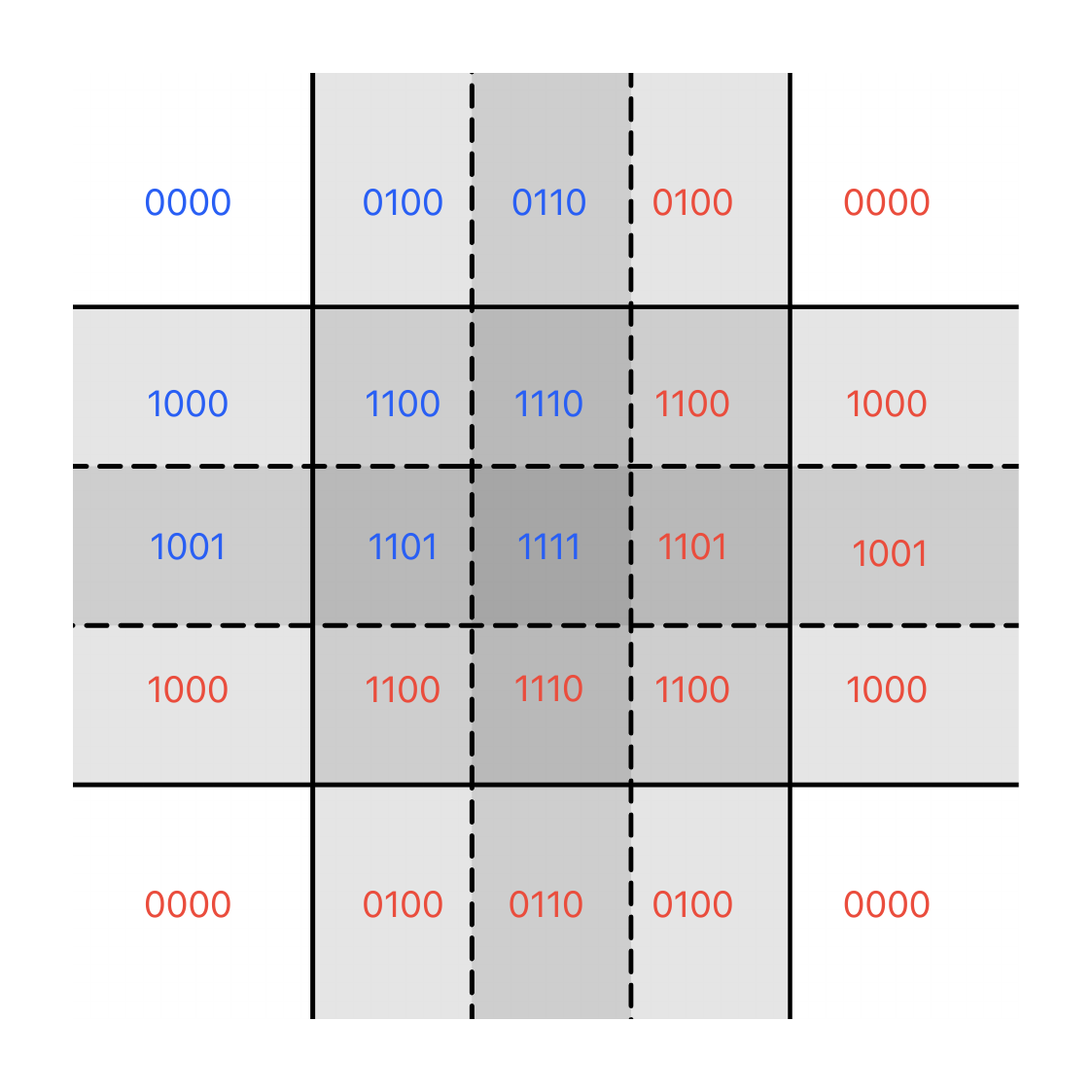}}
   \caption{Equation \eqref{eq:fig:d}}
  \end{subfigure}
  \begin{subfigure}[b]{0.18\linewidth}
   { \includegraphics[width=\linewidth]{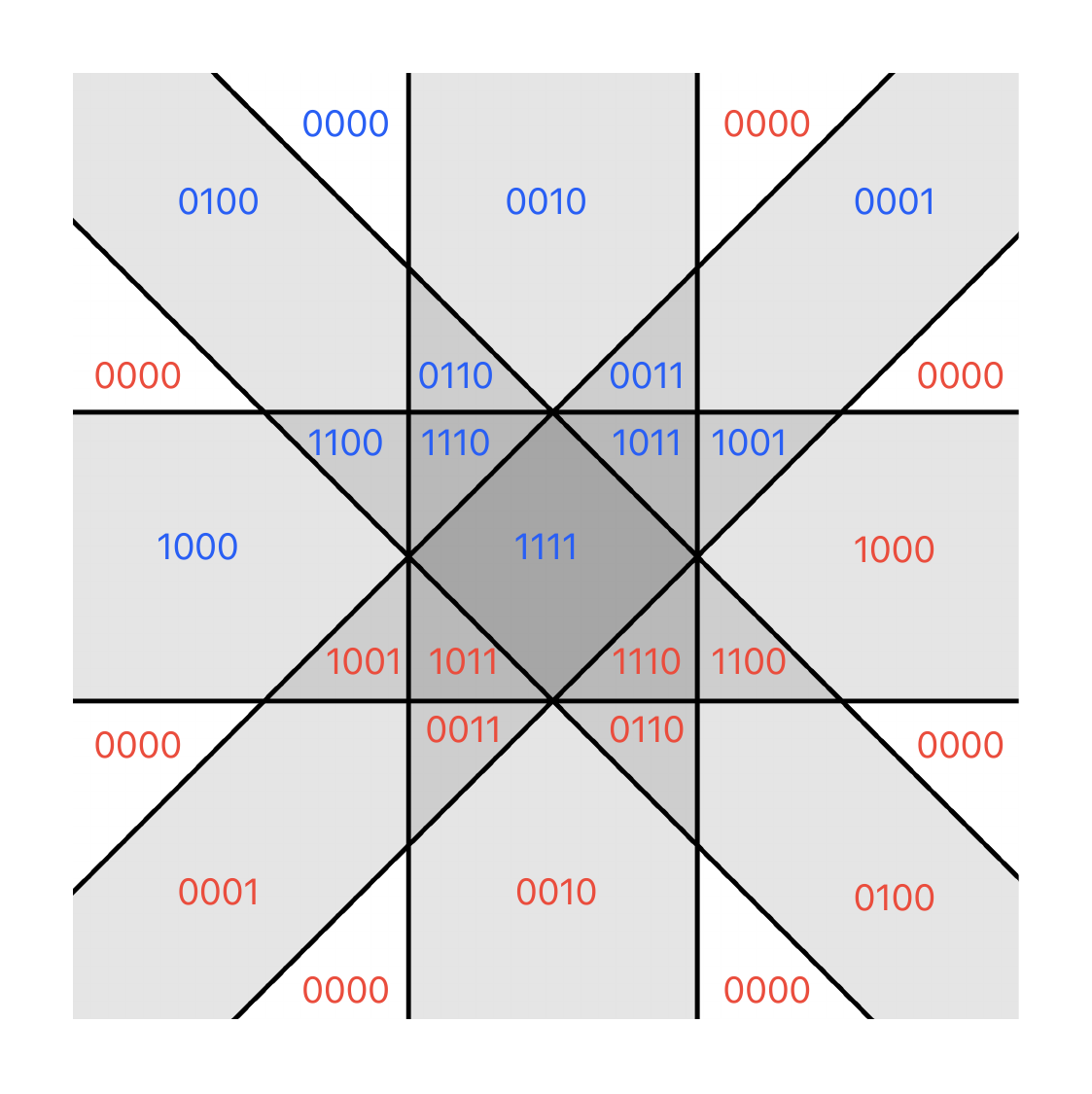}}
   \caption{Equation \eqref{eq:fig:e}}
  \end{subfigure}
  \label{fig:nonlinear_constellations}
  \vspace{-.1in}
    \caption{\textcolor{black}{(a)--(c) The Voronoi Regions for various hybrid beamforming architectures considered in Section \ref{sec:ex1}. (d)-(e)
    Voronoi regions for beamforming architectures with nonlinear analog processing in the example of
    Section \ref{sec:ex1}. Blue-colored binary vectors correspond to regions with unique ADC outputs.}}
  \label{fig:combined_figure}
  \vspace{-0.25in}
\end{figure*}

\begin{figure}[t]
\centering    \includegraphics[width=0.8\linewidth]{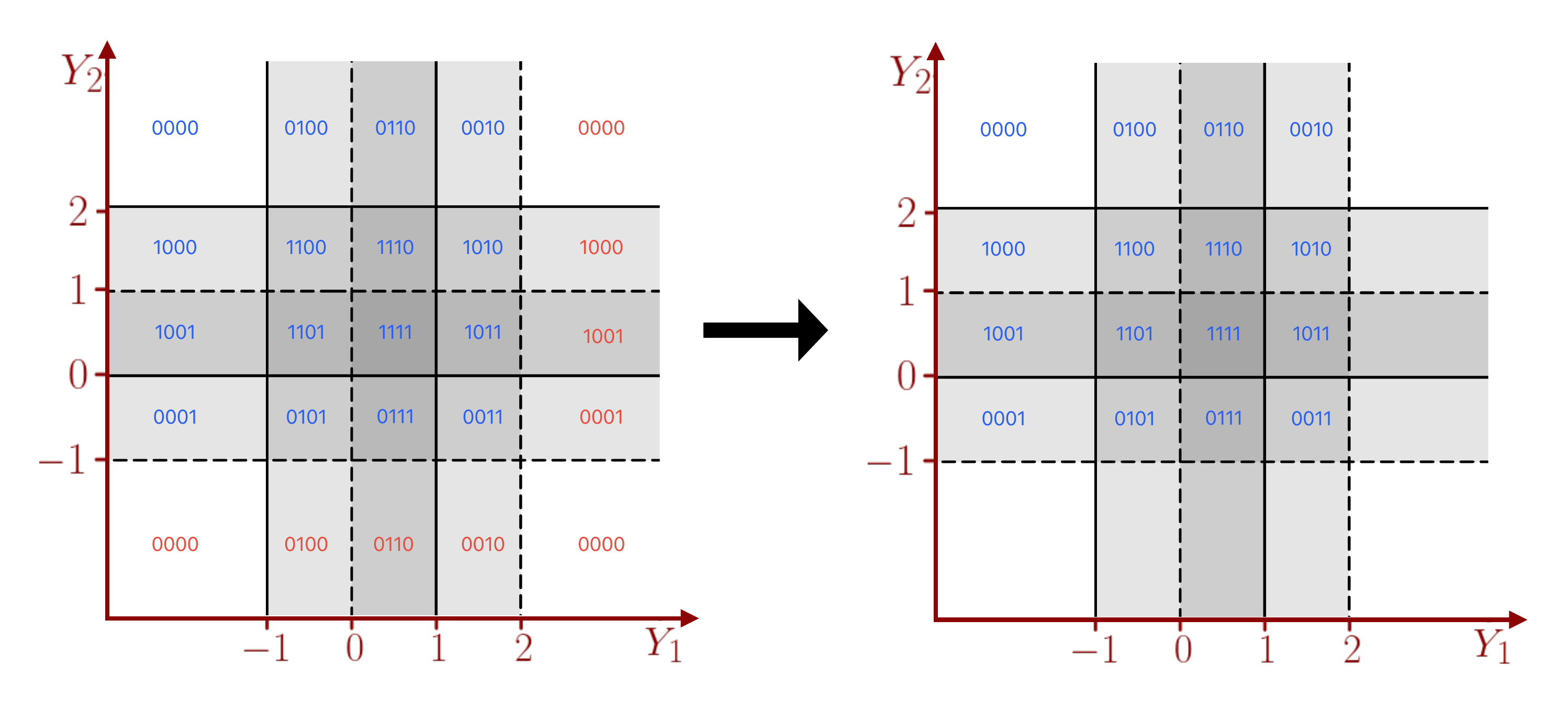}
\vspace{-.1in}
    \caption{\textcolor{black}{The Voronoi regions and ADC outputs produced when using envelope detectors in the example of Equation \eqref{eq:exenv}.}}
    \vspace{-.2in}
\label{fig:unique_regions}
\end{figure}
As discussed in the introduction, reducing the number and resolution of the ADCs leads to a reduction in the resulting channel capacity. This rate-loss is partially unavoidable as the limited number of output bits of the ADCs creates a fundamental bottleneck in the communication channel. For instance, if the receiver is equipped with a single one-bit ADC, then the communication rate cannot exceed one bit per channel-use even in the high SNR regime. More generally, for a receiver equipped with $n_q$ ADCs each with $\ell$ output levels, the communication rate is bounded from above by $n_q\log{\ell}$ bits per channel-use. This implies that the capacity may scale linearly with $n_q$. However, it has been observed in prior works that the channel capacity only scales logarithmically with $n_q$. For instance, in \cite{mo2015capacity}, it was shown that the  high SNR capacity of the single-input multiple-output (SIMO) system with digital beamforming is well-approximated by $\log{(4n_q+1)}$. 
Similarly, the high SNR achievable rate in systems with hybrid beamforming grows logarithmically with $n_q$ \cite{khalili2018mimo}. In this section, we argue that this logarithmic growth, as opposed to the expected linear growth, is due to a phenomenon that we call the \textit{curse of low dimensions}. We first provide an explanation of this phenomenon by evaluating the high-SNR capacity of hybrid beamforming systems equipped with one-bit ADCs. Then, we introduce a mitigating solution, using non-linear analog operators, and show its effectiveness through a numerical example. This intuitive explanation forms the motivation for the rest of the paper. 

\subsection{High SNR Capacity of Hybrid Beamforming Systems}
To provide a high-level description of the curse of low dimensions, let us consider the high SNR capacity of a hybrid beamforming MIMO system equipped with $n_t$ transmit antennas, $n_r$ receive antennas, $n_q$ one-bit ADCs, hybrid beamforming matrix $\mathbf{v}$, and ADC threshold vector $t_{i,1}, i\in [n_q]$.  Furthermore, let us assume that $\mathbf{Y}\approx \mathbf{h}\mathbf{X}$ in the high SNR regime. Hence, the channel output space is equal to the subspace generated by the columns of $\mathbf{h}$, denoted by $\mbox{Span}(\mathbf{h})$.

Each one-bit ADC yields a partition of the output space $\mbox{Span}(\mathbf{h})$ into two subspaces, where one subspace corresponds to points yielding an ADC output equal to zero, and the other corresponds to points yielding ADC output equal to one. To elaborate, recall from Section \ref{sec:arch} that in hybrid beamforming the input to the $i$th ADC is given by $W_i=\mathbf{v}_{i,:}\mathbf{Y}$. Then, the $i$th ADC output is given by 
\begin{align*}
    \widehat{W}_{i}= 
    \begin{cases}
        0\qquad &\text{ if } W_i<t_{i,1},\\
        1&\text{otherwise}
    \end{cases}.
\end{align*}
We define the \textit{decision region} $\mathcal{S}_{i,0}= \{\mathbf{y}| w_i\leq t_i\}, i\in [n_q]$, and its complement as $\mathcal{S}_{i,1}= \mathcal{S}_{i,0}^c$, where $w_i= \mathbf{v}_{i,:}\mathbf{y}$. The  output vector of the ADC module is the binary-valued vector $\mathbf{\widehat{W}}=(\widehat{W}_1,\widehat{W}_2,\cdots,\widehat{W}_{n_q})$, where $\mathbf{Y}\in \mathcal{S}_{i,\widehat{W}_{i}}, i\in [n_q]$. Thus, for a given binary vector $\mathbf{\widehat{w}}\in \{0,1\}^{n_q}$, we have that 
\begin{align*}
    \widehat{\mathbf{W}}= \widehat{\mathbf{w}} \iff 
    \mathbf{Y}\in \bigcap_{i\in [n_q]}\mathcal{S}_{i,\widehat{w}_i}.
\end{align*}
Each non-empty $\bigcap_{i\in [n_q]}\mathcal{S}_{i,\widehat{w}_i}, \mathbf{\widehat{w}}\in \{0,1\}^{n_q}$ is called a Voronoi region. Let $C_h$ be the high SNR capacity. Then, by  the information processing inequality \cite{Cover}, we have that: 
\begin{align*}
    C_h
    = \max_{P_{\mathbf{X}}} I(\mathbf{X};\mathbf{\widehat{\mathbf{W}}}) 
   \leq  \log{|\mathcal{\widehat{W}}|},
\end{align*}
where $\mathcal{\widehat{W}}\!=\! \{\mathbf{\widehat{w}}| \bigcap_{i\in [n_q]}\!\mathcal{S}_{i,\widehat{w}_i}\!\neq\! \varnothing\}\!$ is the set of all Voronoi regions.

Ideally, we expect  $|\mathcal{\widehat{W}}|=2^{n_q}$, so that rate is equal to $n_q$ bits per channel-use. However, this is not the case for low-dimensional spaces. 
To explain further, let the dimension of the subspace $\mbox{Span}(\mathbf{h})$ be denoted by $n_d$. 
Then, the maximum number of partition regions $|\mathcal{\widehat{W}}|$ is equal to the maximum number of regions in an $n_d$-dimensional space when cut by $n_q$ hyperplanes  \cite{winder1966partitions}:
\begin{align}
\label{eq:cuts}
|\mathcal{\widehat{W}}| \leq 2\sum_{i=0}^{n_d-1} { n_q-1 \choose i}.
\end{align}
\noindent
Note that $ { n_q-1 \choose i}\leq (n_q-1)^{i}<n_q^{n_d}$ for $i\in \{0,1,\cdots,n_d-1\}$.
Consequently, for a fixed $n_d$, the maximum achievable rate 
is $\log{|\mathcal{\widehat{W}}|}<n_d\log{n_q}$, and grows logarithmically with  $n_q$. This is also separately observed in prior works \cite{mo2015capacity,khalili2018mimo,khalili2021mimo}.
On the other hand, if $n_d=n_q$, then $\log{|\mathcal{\widehat{W}}|}$ grows linearly with $n_q$. We call this phenomenon the \textit{curse of low dimensions} since it leads to reduced rates when the channel output dimension is smaller than the number of ADCs, i.e., when $n_d< n_q$. 

\subsection{A Mitigating Solution Using Non-Linear Processing}
\label{sec:ex1}
In this work, we argue that non-linear analog processing prior to quantization, using envelope detectors and polynomial operators, mitigates the aforementioned rate-loss. To provide a concrete example, let us consider a high SNR communication scenario with $n_t=n_r=2$ and $n_q=4$ one-bit ADCs at the receiver. Further assume that the channel is full-rank, so that $n_d=2$. 
  The ADCs perform the following comparisons $v_{i,1}Y_1+v_{i,2}Y_2\lessgtr t_i, i\in [4]$. The Voronoi regions
  can be represented as a partition of the two-dimensional Euclidean space using four lines. For instance, in phase-shift keying (PSK) modulation the thresholds are set to $t_i=0, i\in [4]$ and the linear combining coefficients are:
\begin{align*}
    &v_{1,1}=0, \quad v_{1,2}=1,\qquad v_{2,1}=-1,\quad v_{2,2}=1,\\
    &v_{3,1}=1, \quad v_{3,2}=0,\qquad v_{4,1}=1,\quad v_{4,2}=1.
\end{align*} 
As shown in Figure \ref{fig:combined_figure}(a), there are $|\mathcal{\widehat{W}}|=8$ Voronoi regions, as opposed to the ideal 16 regions. Alternatively, Figure \ref{fig:combined_figure}(b) shows the quadrature amplitude modulation (QAM) Voronoi regions, and Figure \ref{fig:combined_figure}(c) represents a general position constellation proposed in \cite{khalili2021mimo}, which achieves the highest number of Voronoi regions possible using hybrid beamforming \cite{khalili2021mimo}. Consequently, PSK achieves a high-SNR rate of at most $\log{8}$, QAM achieves $\log{9}$, and the construction of  \cite{khalili2021mimo} achieves $\log{11}$ bits per channel-use.

We propose to use non-linear analog processing to mitigate the aforementioned rate-loss. To elaborate, assume that the receiver can perform absolute value operations in the analog domain using envelope detectors. In particular, it receives $(Y_1,Y_2)$ and produces the analog signals $W_1=|Y_1|, W_2=|Y_1-1|, W_3=|Y_2|$ and $W_4=|Y_2-1|$. It then feeds each of these four signals to a one-bit ADC with threshold equal to one. That is, the ADC module performs the following comparisons:
\begin{align}
\label{eq:exenv}
|Y_1|\lessgtr1,\quad |Y_1-1|\lessgtr1, \quad |Y_2|\lessgtr1, \quad |Y_2-1|\lessgtr1.
\end{align}
The resulting Voronoi regions are shown in Figure \ref{fig:unique_regions}. Note that the resulting partition is a five-by-five grid in the two-dimensional space, which yields 
 a total of 25 partition elements. However, the ADC outputs are not unique for all of these partition elements. For instance, the top-left partition element, which corresponds to $Y_1\to -\infty$ and $Y_{2}\to \infty$, produces ADC output sequence $\widehat{\mathbf{W}}=(0,0,0,0)$. Similarly, the bottom-right partition element, which corresponds to $Y_1\to \infty$ and $Y_{2}\to -\infty$,  yields ADC output sequence $\widehat{\mathbf{W}}=(0,0,0,0)$. As shown in the figure, there are $|\widehat{\mathcal{W}}|=16$ partition elements with unique ADC outputs, yielding an (optimal) high SNR rate of $\log{16}=4$ bit per channel-use.

 To provide further insights into how the choice of non-linear analog processing functions affects the resulting channel capacity, we have provided two additional examples in Figure \ref{fig:combined_figure}. In particular, Figure \ref{fig:combined_figure}(d) shows the Voronoi regions using zero-bias envelope detectors:
 \begin{align}
     |{Y}_1| \lessgtr 1, \quad |{Y}_2| \lessgtr 1, \quad |{Y}_1| \lessgtr 3,\quad |{Y}_2| \lessgtr 3, 
     \label{eq:fig:d}
 \end{align}
and 
Figure \ref{fig:combined_figure}(e) shows the Voronoi regions resulting from using analog linear combining and zero-bias envelope detectors: 
\begin{align}
|{Y}_1| \lessgtr 1,\quad |{Y}_2| \lessgtr 1, \quad |{Y}_1+{Y_2}| \lessgtr 1, \quad|{Y}_1-{Y}_2| \lessgtr 1.
    \label{eq:fig:e}
\end{align}
It can be observed that the constellation of Figure \ref{fig:combined_figure}(e) achieves a high SNR capacity of $\log{9}$ bits per channel-use, whereas that of \ref{fig:combined_figure}(d) achieves  $\log{12}$ bits per channel-use.



%

\section{SISO Systems Equipped with Envelope Detectors}
\label{sec:env}
In this section, we consider a single-input single-output (SISO) system, where the analog processing operators are restricted to  concatenated sequences of envelope detectors. The (scalar) channel input and output are related through $Y=hX+N$, where $N\sim\mathcal{N}(0,1)$. 
We denote the transmit power by $P_T$ and the ADC power budget by $P_{ADC}$. The maximum number of concatenated envelope detectors is denoted by $\delta$. We represent such a system by the tuple $(P_T,P_{ADC},{h},\delta)$. The resulting channel capacity is represented by $C_{env}(P_T,P_{ADC},{h},\delta)$.
\subsection{Concatenated Sequences of Envelope Detectors}
We model the operation of an envelope detector as an absolute value operation. For instance, an envelope detector with input $y\in \mathbb{R}$ and bias $b\in \mathbb{R}$ outputs $|y-b|$. Sequences of envelope detectors can be concatenated with each other. For instance, consider a concatenated pair of envelope detectors with biases $b_1,b_2\in \mathbb{R}$, respectively. Then, for input $y\in \mathbb{R}$, the output of the concatenated pair is $\big||y-b_1|-b_2\big|$. Due to analog circuit limitations, we assume that a maximum of $\delta\in \mathbb{N}$ envelope detectors can be concatenated.
Consequently, we define the set of implementable analog functions as:
\begin{align*} 
\mathcal{F}^\delta_{env}=\{f(y)=A_s(y,b^s), y\in \mathbb{R}| s\in [\delta], 
    b^s\in \mathbb{R}^s\}, 
 \quad \delta\in \mathbb{N},
\end{align*}
where $A_1(y,b)\triangleq |y-b|, y,b\in \mathbb{R}$ and $A_s(y,b^s)\triangleq  A_1(A_{s-1}(y,b^{s-1}),b_s)= |A_{s-1}(y,b^{s-1})-b_s|, s\in \mathbb{N}$.
 That is, $\mathcal{F}^\delta_{env}$ consists of all functions generated using sequences of $s\leq \delta$ concatenated envelope detectors. Implementability is discussed in Section \ref{sec:cir}. Concatenating large numbers of envelope detectors leads to increased circuit noise and power consumption. Hence, there is a tradeoff between power-consumption, circuit complexity and robustness, and degrees of freedom in implementing analog processing functions.

\subsection{The Quantizer and Its Associated Codebook}
\label{sec:results}
 We call the combination of the envelope detectors and the ADC module the \textit{quantizer}, and represent its operation by a function $\mathbf{q}:\mathbb{R}\to [\ell]^{n_q}$, where $n_q$ is the number of ADCs and $\ell$ is the number of output levels of each ADC, i.e., $\mathbf{q}(y)=\mathbf{a}(\mathbf{f}_a(y))$. This is clarified through the following example. 
 \begin{Example}[\textbf{Quantizer}]
 \label{ex:0}
 Consider a SISO system equipped with two ADCs with three output levels, i.e., $n_q=2$ and $\ell=3$. Assume that a maximum of two envelope detectors can be concatenated with each other, i.e., $\delta=2$. Let $f_{1}(y)=A_2(y,(2,3))=||y-2|-3|$ be the input to the first ADC and $f_{2}(y)= A_2(y,(0,4))=||y|-4|$ the input to the second ADC. Let the ADC thresholds be given by: 
 \begin{align*}
& t_{1,1}= 1, \quad t_{1,2}=2, \quad t_{2,1}= 1, \quad t_{2,2}= 2,
\end{align*}
 The quantization function is $\mathbf{q}(y)=(\widehat{w}_1,\widehat{w}_2)$, where $\widehat{w}_i, i\in\{1,2\}$  are the two ADC outputs. For instance, for $y=-1$ the ADC inputs are
$f_1(-1)=0$ and $f_2(-1)= 3$. 
Since $f_1(-1)<t_{1,1}$ and $f_{2}(-1)>t_{2,2}$, 
the ADC output vector is $\mathbf{q}(-1)= (0,2)$. For general $y\in \mathbb{R}$, the quantizer output $\mathbf{q}(y)$ is shown in the last row of the table in Figure \ref{fig:code}. 
 \end{Example}
 
\begin{Definition}[\textbf{Quantizer}] 
Given the ADC threshold matrix $t^{n_q\times {(\ell-1)}}\!\in \!\mathbb{R}^{n_q\times {(\ell-1)}}$ and implementable analog functions $f_{j}\in\mathcal{F}_{env}^{\delta}, j\in [n_q]$, the quantizer $\mathbf{q}:\mathbb{R}\to [\ell]^{n_q}$ is defined as $\mathbf{q}(\cdot)\triangleq (q_{1}(\cdot)$ $,q_{2}(\cdot),\cdots,q_{n_q}(\cdot))$, where $q_{j}(y)\triangleq k$ if and only if\footnote{For ease of notation, we have defined $t_{j,0}\triangleq -\infty$ and $t_{j,\ell}\triangleq\infty$.} $ f_{j}(y)\in [t_{j,k},t_{j,k+1}], j\in [n_q], k\in\{0,1,\cdots,\ell-1\}$. The associated partition is defined as:
\begin{align}
\mathsf{P}=\{\mathcal{P}_{\mathbf{c}}, \mathbf{c}\in \{0,1,\cdots,\ell-1\}^{n_q}\}- \Phi, 
\label{def:quant}
\end{align}
\text{ where } $\mathcal{P}_\mathbf{c}= \{y\in\mathbb{R}| \mathbf{q}(y)= \mathbf{c}\}, \mathbf{c}\in \{0,1,\cdots,\ell-1\}^{n_q}$.
\end{Definition}


For a quantizer $\mathbf{q}(\cdot)$, we call $y\in \mathbb{R}$ a \textit{point of transition} if the value of $\mathbf{q}(\cdot)$ changes at input $y$, i.e. if it is a point of discontinuity of $\mathbf{q}(\cdot)$. Let $r$ be a point of transition of $\mathbf{q}(\cdot)$. There must exist output vectors $\mathbf{c}\neq \mathbf{c}'$ and $\epsilon>0$ such that $\mathbf{q}(y)=\mathbf{c}, y\in (r-\epsilon,r)$ and $\mathbf{q}(y)=\mathbf{c}', y\in (r,r+\epsilon)$. So, there exists $j\in [n_q]$ and $k\in [\ell-1]$ such that $f_{j}(y)<t_{j,k}, y\in (r-\epsilon,r)$ and $f_{j}(y)\geq t_{j,k}, y\in (r,r+\epsilon)$, or vice versa; so that $r$ is a root of the function $f_{j,k}(y)\triangleq f_{j}(y)-t_{j,k}$.
Note that there are a total of at most $(\ell-1)  n_q2^\delta$ roots for the collection of functions $f_{j,k}(\cdot)$.  
Let $r_1,r_2,\cdots,r_{\gamma}$ be the sequence of roots of $f_{j,k}(\cdot), j\in [n_q], k\in [\ell-1]$ (including repeated roots), written in non-decreasing order, where $\gamma\leq (\ell-1)  n_q2^\delta$. Let $\mathcal{C}=(\mathbf{c}_0,\mathbf{c}_1,\cdots, \mathbf{c}_{\gamma})$ be the corresponding quantizer outputs, i.e. $\mathbf{c}_{i-1}= \lim_{y\to r_i^-}\mathbf{q}(y), i\in [\gamma]$ and $\mathbf{c}_{\gamma}=\lim_{y\to\infty}\mathbf{q}(y)$. We call $\mathcal{C}$ the \textit{code} associated with the quantizer. Note that the associated code is an ordered set of vectors.
 Each $\mathbf{c}_i= (c_{i,1},c_{i,2},\cdots,c_{i,n_q}), i\in \{0,1,\cdots,\gamma\}$ is called a codeword. The size of the code $|\mathcal{C}|$ is defined as the number of unique vectors in $\mathcal{C}$. For a fixed $j\in [n_q]$, the transition count of position $j$ is the number of codeword indices where the value of the $j$th element changes, and it is denoted by $\kappa_j$, i.e., $\kappa_j\triangleq \sum_{k=1}^{\gamma}\mathbbm{1}(c_{k-1,j}\neq c_{k,j})$.
It is straightforward to see that $|\mathsf{P}|=|\mathcal{C}|$ since both cardinalities are equal to the number of unique outputs the quantizer produces. 
The following example clarifies the definitions given above. 
 
 \begin{figure}[t]
\centering 
\includegraphics[width=0.4\textwidth]{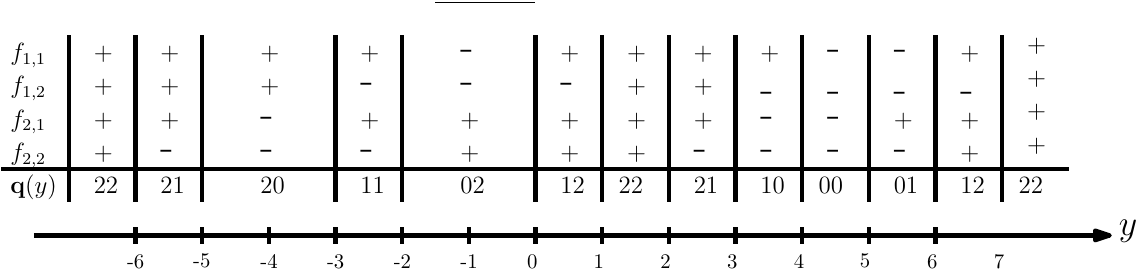}
\vspace{-.05in}
\caption{The quantizer outputs in Examples \ref{ex:0} and \ref{Ex:1}. The first four rows show the sign of $f_{j,k}, j\in [n_q],k\in [\ell-1]$ for $y$ within each interval, and the last row is the quantizer output.
\vspace{-.05in}
}
\vspace{-.2in}
\label{fig:code}
\end{figure}
\begin{Example}[\textbf{Associated Code}]
\label{Ex:1}
Consider the setting of Example \ref{ex:0}. The functions $f_{j,k}, j\in [n_q], k\in [\ell-1]$ are:
\begin{align*}
    &f_{1,1}(y)= \big||y-2|-3\big|-1, \qquad f_{1,2}(y)= \big||y-2|-3\big|-2,\\
    &f_{2,1}(y)= \big||y|-4\big|-1, \qquad f_{2,2}(y)= \big||y|-4\big|-2.\\
\end{align*}

The ordered root sequence of the quantizer is 
$(r_1,r_2,\cdots,r_{13})=
(-6,-5,-3,-2,0,1,2,3,4,5,6,7)$. The 
associated partition is:
\begin{align*}
   & \mathsf{P}= \Big\{(-\infty,-6), [-6,-5),[-5,-3),
 [-3,-2),[-2,0),[0,1),[1,2)
 \\&\qquad [2,3),[3,4),[4,5), [5,6),[6,7),[7,\infty)\Big\}.
\end{align*}
The associated code is given by $22,21,20,11,02,12,$ $22,21,10,00,01,12,22$. This is shown in Figure \ref{fig:code}. The size of the code is $|\mathcal{C}|=13$. The high SNR capacity of a channel using this quantizer at the receiver is $\log{|\mathsf{P}|}=\log{|\mathcal{C}|}=\log{13}$.
\end{Example}
\label{sec:prelim}

\subsection{Single Envelope Detector and One-bit Quantization}
As a first step, we investigate scenarios with one-bit ADCs and  
a single envelope detector per ADC, i.e., 
$\ell=2$ and $\delta=1$. 
Since $\delta=1$, each $f_{j,1}(y)$ is of the form $|y-b_j|$ for some $b_j\in \mathbb{R}$.
 Given an ADC threshold $t_{j,1}>0$ the roots of $f_{j,1}(y)-t_{j,1}$ are equal to $b_j \pm t_{j,1}$. So, there are a total of at most $2n_q$ roots for the quantizer.
Then, the number of Voronoi regions of the quantization, which are segments of the real line separated by these roots, cannot be greater than $2n_q+1$. On the other hand, the first and last  segments always have the same ADC outputs (for an example, refer to Figure \ref{fig:code}). Consequently, the high SNR capacity is upper-bounded by $\log{2n_q}= \log{2\floor{\frac{P_{ADC}}{2\alpha}}}$.  The following proposition shows that this upper-bound is achievable. The proof is provided in \cite{Shirani_JSTSP_2022}. 
\begin{Proposition}
\label{Prop:1}
Let $h\neq0$, $P_{ADC}>0$, $\delta=1$, and $\ell=2$. Then, 
\begin{align}
\label{eq:prop:1}
 \lim_{P_T\to \infty}   C_{env}(P_T,P_{ADC}, h,\delta)= 1+\log{\lfloor\frac{P_{ADC}}{2\alpha}\rfloor}, 
\end{align}
where $2\alpha$ is the unit of power consumption for a one-bit ADC.
\end{Proposition}
The proof is provided in Appendix \ref{App:Prop1}. 

To compare this result with the conventional approach, where analog processing is restricted to linear operators, note that without using envelope detectors and with $n_q$ one-bit ADCs one can partition the real line into at most $n_q+1$ Voronoi regions. Consequently, the capacity of such a system at high SNR is $\log{(n_q+1)}$. In comparison, when using envelope detectors the capacity is $\log{2n_q}$ which yields an improvement. 


The following provides a computable expression for the capacity under general assumptions on SNR. 
\begin{Theorem}
\label{th:1}
Consider a system parametrized by $(P_T,P_{ADC},h,\delta)$, where $P_T,P_{ADC}>0, h\neq 0, n_q>1$, and $\delta=1,\ell=2$. Then, the capacity is given by:
\begin{align}
\label{eq:th:1}
    C_{env}(P_T,P_{ADC},h,\delta)=\max_{\mathbf{x}\in \mathbb{R}^{2n_q+1}} \max_{P_{X}\in \mathcal{P}_{\mathbf{x}}(P_T)} \max_{\mathbf{t}\in \mathbb{R}^{2n_q}} I(X;\widehat{Y}),
\end{align}
where $n_q= \floor{\frac{P_{ADC}}{2\alpha}}$,
$\widehat{Y}= q(hX+N)$, $N\sim \mathcal{N}(0,1)$, $\mathcal{P}_{\mathbf{x}}(P_T)$ is the set of probability distributions defined on $\{x_1,x_2,\cdots,x_{2n_q+1}\}$ such that $\mathbb{E}(X^2)\leq P_T$,
and $q(y)=k$ if $y\in [t_{k},t_{k+1}], k\in [2n_q]$ and $q(y)=0$ if $y>t_{2n_q}$ or $y<t_{1}$.
\end{Theorem}
The proof is provided in Appendix \ref{App:th:1}. 
The theorem shows that, similar to the high-SNR regime, for low-SNRs the effect of using envelope detectors without concatenation is (almost) equivalent\footnote{Note that the effect is not exactly equal to doubling the number of ADCs as the first and last quantization regions do not produce unique ADC outputs when using envelope detectors as can be observed in Figure \ref{fig:code}.} to increasing the number of ADCs from $n_q$ to $2n_q$. 

\subsection{Concatenated Envelope Detectors and Few-bit ADCs}
Next, we consider systems with $\delta>1$ and $\ell>2$. Recall that for $\delta>1$, each $f_j(y),j \in [n_q]$ is of the form $f_j(y)=A_\delta(y,b_j^\delta)=|A_{\delta-1}(y,b_j^{\delta-1})-b_{j,\delta}|$, where $A_1(y,b_{j,1})=|y-b_{j,1}|$. 
To evaluate the resulting channel capacity, we need to characterize the number and shape of Voronoi regions which can be produced using a quantizer consisting of concatenated envelope detectors and few-bit ADCs. To this end, we first introduce the notion of non-degenerate bias vectors (Proposition \ref{prop:2.5}). Then, for envelope detectors with non-degenerate biases, we characterize some key properties of the associated code and the roots of the quantization function (Propositions \ref{prop:1.5} and \ref{Prop:2}). Subsequently, for any given sequence of roots satisfying the aforementioned properties, we provide a code construction mechanism which has that root sequence (Proposition \ref{Prop:3}). This enables us to derive the capacity expression for the communication systems under consideration (Theorem \ref{th:4}).

Some (degenerate) bias vectors lead to redundant envelope detectors. The following example clarifies the notion. 
\color{black}
\begin{Example}[\textbf{Degenerate Bias Vector}]
\label{ex:2}
    Let $n_q=1, \ell=3, \delta=2$, and consider the thresholds $t_{1,1}=1$ and $t_{1,2}=2$. Given a bias vector $(b_1,b_2)$, the associated analog function is $f_{j}(y)= ||y-b_1|-b_2|$. The ADC output is
    \begin{align*}
        q(y)=
        \begin{cases}
            0 \qquad & \text{ if } ||y-b_1|-b_2|<1\\
            1 & \text{ if } 1\leq ||y-b_1|-b_2|< 2
\\            2 & \text{ if } 2\leq ||y-b_1|-b_2|
        \end{cases}.
    \end{align*}
    Note that if $b_2-1<0$, then this would be equivalent with:
        \begin{align*}
        q(y)=
        \begin{cases}
            0 \qquad & \text{ if } |y-b_1|<1+b_2\\
            1 & \text{ if } 1+b_2\leq |y-b_1|<2+b_2
\\            2 & \text{ if } 2+b_2\leq |y-b_1|
        \end{cases}.
    \end{align*}
 In this case, the second envelope detector does not affect the quantization process and can be omitted without change in quantizer output, i.e., the input to the corresponding absolute value is always positive, so it can be removed.
\end{Example}
 The following proposition characterizes the necessary and sufficient conditions for non-degeneracy of bias vectors.
 \begin{Proposition}[\textbf{Non-Degenerate Bias Vectors}]
 \label{prop:2.5}
    Let the threshold vector corresponding to the $j$th ADC be $t^{\ell-1}$, where $j\in [n_q]$. The bias vector $\mathbf{b}$ of the corresponding analog operator $f_j(\cdot)$ is non-degenerate if and only if:
    \begin{align}
    \label{eq:degen}
        0< t_1+\sum_{i=2}^\delta(-1)^{p_i}b_i, \quad \forall p_i \in \{-1,1\}. 
    \end{align} 
 \end{Proposition}
 The proof follows by noting that from definition $0<t_1<t_2<\cdots<t_{\ell-1}$ so that Equation \eqref{eq:degen} guarantees $ 0< t_k+\sum_{i=2}^\delta(-1)^{p_i}b_i$ for all $k\in [\ell-1]$, and 
 is thus sufficient to ensure non-degeneracy. 

In the next step, we introduce the notion of a fully-symmetric vector which will be used in deriving properties of roots of quantizers with non-degenerative bias vectors.
\begin{Definition}[\textbf{Fully-Symmetric Vector}] A vector $\mathbf{m}=(m_1,m_2,\cdots,m_{2^n})$ is called symmetric if $m_i+m_{2^n-i}=m_j+m_{2^n-j}, i,j\in [2^n-1]$.
$\mathbf{m}$ is called fully-symmetric if it is symmetric 
and the vectors  $(m_1,m_2,\cdots,m_{2^{n-1}})$ and $(m_{2^{n-1}+1},m_{2^{n-1}+2},\cdots,m_{2^{n}})$ are both fully-symmetric for $n>2$ and symmetric for $n=2$.  
\end{Definition}
For instance, the vector $\mathbf{m}=(-7,-6,-5-4,4,5,6,7)$ is fully symmetric since it is symmetric and $(-7,-6,-5-4)$
 and $(4,5,6,7)$ are both symmetric. 
\textcolor{black}{As observed in the proof of Theorem \ref{th:1}, a critical step in deriving the channel capacity resulting from a specific constellation of envelope detectors and ADCs is to characterize the number and shape of the Voronoi regions of the quantization operation resulting from the constellation. To this end, the following proposition provides useful properties of the roots which in turn determine the resulting Voronoi regions.}
 \begin{Proposition}
\label{prop:1.5}
     Consider a quantizer $\mathbf{q}(\cdot)$ with threshold matrix $t^{n_q\times (\ell-1)}\in \mathbb{R}^{n_q\times (\ell-1)}$,
and analog processing functions $f_{j}(\cdot)\in \mathcal{F}_{env}^{\delta}, j\in n_q$, such that 
the corresponding bias vectors are non-degenerate and $f_{j,k}(\cdot)\triangleq f_{j}(\cdot)- t_{j,k}, j\in [n_q], k\in [\ell-1]$ do not have repeated roots. Let $r_1,r_2,\cdots,r_{\gamma}$ be the increasing sequence of roots, where $\gamma\triangleq (\ell-1)n_q2^{\delta}$. Then, there exists a partition $\{\mathcal{P}_{j,k}, j\in [n_q], k\in [\ell-1]\}$ of $[\gamma]$ such that 
\\1) $|\mathcal{P}_{j,k}|=2^{\delta}, j\in [n_q], k\in [\ell-1]$.
\\2)  For $j\in [n_q], k\in [\ell-1]$, let $\mathcal{P}_{j,k}=\{i_1,i_2,\cdots, i_{2^{\delta}}\}$, where $i_j<i_{j'}$ for $j<j'$. The vector $(r_{i_1}, r_{i_2}, \cdots, r_{i_{2^\delta}})$ is fully-symmetric, 
\\3) For all $j\in [n_q], k,k'\in [\ell-1]$, we have $r_{i_t}-r'_{i_t}= r_{i'_t}-r'_{i'_t}, i_t,i'_t\in [2^\delta]$, where $r_{i_t},r_{i'_t}\in \mathcal{P}_{j,k}$ and $r'_{i_t},r'_{i'_t}\in \mathcal{P}_{j,k'}$.
 \end{Proposition}

The proof follows by taking each $\mathcal{P}_{j,k}$ to be the ordered set of roots of $f_{j,k}(\cdot)$ for a given $j\in [n_q], k\in [\ell-1]$ and using properties of the absolute value operator. The following proposition states several useful properties for the code associated with a quantizer $\mathbf{q}(\cdot)$.
\begin{Proposition}[\textbf{Properties of the Associated Code}] 
\label{Prop:2}
Consider a quantizer $\mathbf{q}(\cdot)$ with threshold matrix ${t}^{n_q\times (\ell-1)}\in \mathbb{R}^{n_q\times (\ell-1)}$ such that $0<t_{i,j}<t_{i,j'}, i\in [n_q], j,j'\in [\ell-1], j<j'$,
and analog processing functions $f_{j}(\cdot), j\in n_q$, such that 
the corresponding bias vectors are non-degenerate and $f_{j,k}(\cdot)\triangleq f_{j}(\cdot)- t_{j,k}, j\in [n_q], k\in [\ell-1]$ do not have repeated roots. The associated code $\mathcal{C}$ satisfies the following:
\begin{enumerate}[leftmargin=*]
\color{black}
\item The number of codewords in $\mathcal{C}$ is equal to $\gamma+1$, where $\gamma\triangleq (\ell-1) n_q2^\delta $, i.e. $\mathcal{C}=(\mathbf{c}_0, \mathbf{c}_1,\cdots, \mathbf{c}_{\gamma})$.
    \item All elements of the first codeword $\mathbf{c}_0$ are equal to $\ell-1$, i.e. $c_{0,j}=\ell-1, j\in [n_q]$.
    \item Consecutive codewords differ in only one position, and their $L_1$ distance is equal to one, i.e. $\sum_{j=1}^{n_q}|c_{i,j}-c_{i+1,j}|=1, i\in \{0,1,\cdots,\gamma-1\}$.
    \item  The transition count at every position is $\kappa_j= \frac{\gamma}{n_q}= (\ell-1)2^\delta, j\in [n_q]$.
    \item For $j\in [n_q]$, let $i_1,i_2,\cdots, i_{\kappa_j}$ be the non-decreasingly ordered indices of codewords, where the $j$th element has value-transitions. Then, the sequence $(c_{i_1,j},c_{i_2,j},\cdots,c_{i_{\kappa_j},j})$ is periodic, in each period it takes all values between $0$ and $\ell-1$, and  $|c_{i_k,j}-c_{i_{k+1},j}|=1, k\in [\kappa_j-1]$ holds. 
    \item  $|\mathcal{C}|\leq min(\ell^{n_q}, (\ell-1)n_q2^{\delta})$.
    
\end{enumerate}

\end{Proposition}
\color{black}
 Proposition \ref{Prop:2} is an extension of the properties shown in the proof of Theorem \ref{th:1}. We provide a brief justification of each property in the following. Property 1 follows by the fact that the number of codewords in $\mathcal{C}$ is one more than  the number of roots of $f_{j,k}, j\in [n_q], k\in [\ell-1]$. 
 Property 2 follows from the fact that as $y\to \pm\infty$, its absolute value becomes asymptotically large.  
 Property 6 follows from the facts that i) $n_q$ ADCs with $\ell$ output levels cannot output more than $\ell^{n_q}$ distinct outputs, ii) the number of roots cannot exceed $(\ell-1)n_q2^{\delta}$, and iii) the first and last decision regions have the same ADC outputs. Properties 3 and 5 follow by the fact that each root of  $f_{j,k}, j\in [n_q], k\in [\ell-1]$ corresponds to a value transition in the output of exactly one of the ADCs (since the roots are not repeated) and at each transition the value changes either one unit up or down since in the input crosses one threshold at a time and its value is changed continuously. Property 4 follows by the fact that the transition count at each position is equal to the number of roots of $f_{j,k}, k\in [\ell-1]$ for a fixed $j\in [n_q]$.
\color{black}
 
As a step towards characterizing capacity when $\ell>2,\delta>1$, we first study the capacity region for systems with one-bit ADCs, i.e.,  $\ell=2,\delta>1$. To this end, we prove two useful propositions. The first one shows that given an ordered set $\mathcal{C}$ satisfying the properties in Proposition \ref{Prop:2} and a sequence of real numbers $(r_1,r_2,\cdots,r_{\gamma})$ satisfying the properties in Proposition \ref{prop:1.5}, one can always construct a quantizer whose associated code is equal to $\mathcal{C}$ and whose  roots sequence is $(r_1,r_2,\cdots,r_{\gamma})$. The second proposition provides conditions under which there exists a code satisfying the properties in Proposition \ref{Prop:2}. \textcolor{black}{The proof follows by construction, using techniques in balanced and locally balanced Gray codes \cite{bykov2016locally}. The constructive proof is detailed in \cite{Shirani_JSTSP_2022}.}
Combining the two results allows us to characterize the necessary and sufficient conditions for existence of quantizers with desirable properties.

For a given code $\mathcal{C}$, we denote $\xi_k, k\in [\gamma]$ as the bit position which is different between $\mathbf{c}_{k-1}$ and $\mathbf{c}_{k}$. We define
the transition sets $\mathcal{I}_j\triangleq\{k| \xi_k=j\}, j\in [n_q]$. Note from Property 5 in Proposition \ref{Prop:2}, we have $|\mathcal{I}_j|=\kappa_j= 2^{\delta}$.
\color{black}
\begin{Proposition}[\textbf{Quantizer Construction}]
\label{Prop:3}
Let $\ell=2, n_q\in \mathbb{N}$, and $\delta>1$ such that $n_q2^{\delta}\leq 2^{n_q}$, and consider an ordered set $\mathcal{C}\subset \{0,1\}^{n_q}$ satisfying properties 1-5 in Proposition \ref{Prop:2}, and a sequence of increasing real numbers $r_1,r_2,\cdots, r_{\gamma}$, where $\gamma= n_q2^{\delta}$,  such that $(r_{i_k}, k\in \mathcal{I}_j)$ is fully-symmetric for all $j\in [n_q]$, where $\mathcal{I}_j$ are the transition sets of $\mathcal{C}$. Then, there exists a quantizer $\mathbf{q}(\cdot)$ with associated functions $f_{j}(\cdot)\in \mathcal{F}^{\delta}_{env}, j\in [n_q]$ such that its code is $\mathcal{C}$, and $r_1,r_2,\cdots, r_{\gamma}$ is the non-decreasing sequence of its associated roots. 
\end{Proposition}
\color{black}

\begin{Proposition}(\textbf{Code Construction})
\label{Prop:6}
Let $\ell=2$, $n_q\in \mathbb{N}$, and $\kappa_1$, $\kappa_2,\cdots,\kappa_{n_q}$ be even numbers such that $|\kappa_j-\kappa_{j'}|\leq 2, j,j'\in [n_q]$. Then, there exists a code $\mathcal{C}$ with transition count at position j equal to $\kappa_j, j\in [n_q]$ satisfying properties 1, 2, 3, and 5 in Proposition \ref{Prop:2} such that $|\mathcal{C}|=\min\{2^{n_q}, \sum_{j=1}^{n_q}\kappa_j\}$. Particularly, if $\kappa_j=2^\delta, j\in [n_q]$, then there exists $\mathcal{C}$ with $|\mathcal{C}|= \min\{2^{n_q}, n_q2^\delta\}$ satisfying properties 1-5 in Proposition \ref{Prop:2}.
\label{Prop:4}
\end{Proposition}
The proof of Propositions \ref{Prop:3} and \ref{Prop:4} are provided in Appendix \ref{App:Prop:5} and Appendix \ref{App:Prop:6}, respectively.
Using Propositions \ref{Prop:3} and \ref{Prop:4}, we characterize the channel capacity for $\ell=2$ and $\delta>1$. 
Let us define $\gamma'\triangleq \min(2^{n_q}, n_q2^\delta )$ and the set $\mathcal{T}_{\gamma'}\subseteq \mathbb{R}^{\gamma'}$ as the set of  sequences of increasing real numbers $r_1,r_2,\cdots, r_{\gamma'}$ for which  there exists a partition $\{\mathcal{I}_j, j\in[n_q]\}$ of $[\delta]$
such that $(r_{i_k}, k\in \mathcal{I}_j)$ is fully-symmetric for all $j\in [n_q]$, and there exists a code satisfying Properties 1-5 in Proposition \ref{Prop:2} whose transition sets are equal to $\mathcal{I}_j, j\in [n_q]$ and which has exactly one repeated codeword, i.e., only the first and last codewords are repeated. The following theorem characterizes the channel capacity.
\color{black}
\begin{Theorem}[\textbf{SISO Capacity with One-bit ADCs and Multi-level Envelope Detectors}]
\label{th:2}
Consider a system parametrized by $(P_T,P_{ADC},h,\delta)$, where $P_T,P_{ADC}>0, h\in \mathbb{R}$, $\delta>1$, and let $\ell=2$. Then, the capacity is given by:
\begin{align}
\label{eq:th:2}
    C_{env}(P_T,P_{ADC},h,\delta)=
   \max_{\mathbf{x}\in \mathbb{R}^{ \gamma'+1}} \max_{P_{X}\in \mathcal{P}_{\mathbf{x}}(P_T)} \max_{\mathbf{t}\in \mathcal{T}_{\gamma'}} I(X;\widehat{Y}),
\end{align} 
where $n_q=\floor{\frac{P_{ADC}}{2\alpha}}$, 
$\widehat{Y}= q(hX+N)$, $N\sim \mathcal{N}(0,1)$, $\mathcal{P}_{\mathbf{x}}(P_T)$ is the set of distributions on  $\{x_1,x_2,\cdots,x_{\gamma'+1}\}$ such that $\mathbb{E}(X^2)\leq P_T$,
and $q(y)=k$ if $y\in [t_{k},t_{k+1}], k\in \{1,\cdots,\gamma'-1\}$ and $q(y)=0$ if $y>t_{\gamma'-1}$ or $y<t_{1}$.
\end{Theorem}
\color{black}
\textcolor{black}{To provide insights into the implications of the above theorem, let us consider SISO systems where the ADC thresholds are restricted to be fully-symmetric. Then, the theorem implies that  the effect of using $\delta$ concatenated envelope detectors is (almost) equivalent to increasing the number of ADCs exponentially from $n_q$ to $2^{\delta}n_q$. }
The proof follows by similar arguments as in the proof of Theorem \ref{th:1}. The converse follows from Proposition \ref{Prop:2} Item 4). Achievability follows from Proposition \ref{Prop:4}.

The region given in Theorem \ref{th:2} is difficult to analyze since finding the set $\mathcal{T}_{\gamma'}$ may be computationally complex. Inner bounds to the achievable region may be numerically derived by assuming additional symmetry restriction such as uniform quantization restrictions. This is studied in more detail in the numerical evaluations provided in Section \ref{sec:num}.


For scenarios with $\ell>2$ and $\delta>1$,
let us define $\gamma_\ell\triangleq \min(\ell^{n_q}, (\ell-1)n_q\ell^\delta)$ and the set $\mathcal{T}_{\gamma_\ell}\subseteq \mathbb{R}^{\gamma_\ell}$ as the set of  sequences of increasing real numbers $r_1,r_2,\cdots, r_{\gamma_{\ell}}$ satisfying the properties in Proposition \ref{prop:1.5}, for which there exists a code $\mathcal{C}$ satisfying Properties 1-5 in Proposition \ref{Prop:2} such that  i)
the sets $\mathcal{P}_{j,k}, j\in [n_q], k\in [\ell-1]$ in Proposition \ref{prop:1.5} are the indices of the codewords of $\mathcal{C}$ which have transition to or from value $k$ in their $j$th element, and ii) $\mathcal{C}$ has exactly one repeated codeword, i.e., only the first and last codewords are repeated. 
The following theorem characterizes the channel capacity. The proof follows from Propositions \ref{Prop:2} and \ref{Prop:4} similar to the arguments given in the proof of Theorem \ref{th:1}.

\color{black}
\begin{Theorem}[\textbf{SISO Capacity with Few-bit ADCs and Multi-level Envelope Detectors}]
\label{th:4}
Consider a system parametrized by $(P_T,P_{ADC},h,\delta)$, where $P_T,P_{ADC}>0, h\in \mathbb{R}$, and $\delta\in \mathbb{N}$. Then, 
\begin{align}
\label{eq:th:4}
\!\!\! C_{env}(P_T,P_{ADC},h,\delta) =
 \max_{\substack{(n_q,\ell):\\ \alpha\ell n_q\leq P_{ADC}}}
 \max_{\mathbf{x}\in \mathbb{R}^{\gamma_{\ell}+1}} \max_{P_{X}\in \mathcal{P}_{\mathbf{x}}} \max_{\mathbf{t}\in \mathcal{T}_{\gamma_{\ell}}} I(X;\widehat{Y}),
\end{align}
where 
$\widehat{Y}= q(hX+N), N\sim \mathcal{N}(0,1)$, $\mathcal{P}_{\mathbf{x}}(P_T)$ is the set of distributions on  the alphabet $\{x_1,x_2,\cdots,x_{\gamma_{\ell}+1}\}$ such that $\mathbb{E}(X^2)\leq P_T$,
and $q(y)=k$ if $y\in [t_{k},t_{k+1}], k\in [\gamma_{\ell}-1]$ and $q(y)=0$ if $y>t_{\gamma_{\ell}-1}$ or $y<t_{1}$.\end{Theorem}
\color{black}
\color{black}
\section{SISO Systems with Polynomial Operators}
\label{sec:poly}
 In this section, we evaluate the channel capacity when polynomial operators are used instead of envelope detectors. 
\subsection{Analog Polynomial Operators}  
We restrict the set of implementable nonlinear analog functions to the set of polynomials of degree at most $\delta$, that is: 
 \[\mathcal{F}^{\delta}_{poly}\triangleq \{f(y)=\sum_{i=0}^\delta{b_i}y^i, y\in \mathbb{R}| b_i \in \mathbb{R}, i\in [\delta]\},\quad \delta\in \mathbb{N},\]
 For $P_T,P_{ADC}>0$ and $h\neq 0$, 
 the resulting channel capacity is represented by $C_{poly}(P_T,P_{ADC},{h},\delta)$.
Similar to the case with envelope detectors, here the limitation on the degree of the polynomial functions is determined by limitations in the corresponding circuits design and is studied in Section \ref{sec:cir}. We define the quatnizer and associate code in a similar manner as in the previous section. This is clarified through the following example. 
\color{black}


\begin{Example}[\textbf{Associated Code}]
\label{Ex:1p}
Let $n_q=\delta=2$ and $\ell=3$ and consider a quantizer characterized by polynomials
$f_{1}(y)=y^2+2y$ and $f_{2}(y)= y^2+3y, y\in \mathbb{R}$, and thresholds
\begin{align*}
& t_{1,1}= 0, \quad t_{1,2}=3, \quad t_{2,1}= 10, \quad t_{2,2}= 18,
\end{align*}
We have:
\begin{align*}
    & f_{1,1}(y)= y^2+2y, \quad f_{1,2}(y)= y^2+2y-3\\
    & f_{2,1}(y)=y^2+3y-10, \quad  f_{2,2}(y)= y^2+3y-18.
\end{align*}
The ordered root sequence is $(r_1,r_2,\cdots,r_8)=
(-6,$ $-5,-3,-2,0,1,2,3)$. The 
associated partition is:
\begin{align*}
   & \mathsf{P}= \Big\{(-\infty,-6), [-6,-5),[-5,-3),
 [-3,-2),[-2,0),
 \\&\qquad [0,1),[1,2),[2,3),[3,\infty)\Big\}.
\end{align*}
The associated code is given by $22,21,20,10,$ $00,10,20,21,22$. The size of the code is $|\mathcal{C}|=5$. The high SNR capacity of a channel using this quantizer at the receiver is $\log{|\mathsf{P}|}=\log{|\mathcal{C}|}=\log{5}$.
\end{Example}

The following provides the SISO capacity with few-bit ADCs and polynomial operators. The proof is provided in \cite{Shirani_JSTSP_2022}.
\color{black}
\begin{Theorem}[\textbf{SISO Capacity with Few-bit ADCs and Polynomial Analog Operators}]
\label{th:5}
    Consider a system parametrized by $(P_T,P_{ADC},h,\delta)$, where $P_T,P_{ADC}>0, h\neq 0$. For even-valued $\delta$, we have:
\begin{align}
\label{eq:th:5}
    C_{poly}(P_T,P_{ADC},h,\delta)= \max_{\substack{(n_q,\ell):\\ \alpha \ell n_q\leq P_{ADC}}}\max_{\mathbf{x}\in \mathbb{R}^{ \gamma+1}} \max_{P_{X}\in \mathcal{P}_{\mathbf{x}}(P_T)} \max_{\mathbf{t}\in \mathbb{R}^{\gamma}} I(X;\widehat{Y}),
\end{align}
and for odd-valued $\delta$, we have:
\begin{align}
\label{eq:th:3}
&\max_{(n_q,\ell): \alpha \ell n_q\leq P_{ADC}}
 \max_{\mathbf{x}\in \mathbb{R}^{\gamma+1}} \max_{P_{X}\in \mathcal{P}_{\mathbf{x}}(P_T)} \max_{\mathbf{t}\in \mathbb{R}^{\gamma}} I(X;\widehat{Y})  \leq C_{poly}(P_T,P_{ADC},h,\delta)
 \\&\nonumber\qquad \qquad  \leq 
 \max_{(n_q,\ell): \alpha \ell n_q\leq P_{ADC}}\max_{\mathbf{x}\in \mathbb{R}^{\gamma'+1}} \max_{P_{X}\in \mathcal{P}_{\mathbf{x}}(P_T)} \max_{\mathbf{t}\in \mathbb{R}^{\gamma'}} I(X;\widehat{Y'}),
\end{align}
 where $n_q= \floor{\frac{P_{ADC}}{2\alpha}}$, $\gamma\triangleq \min(\ell^{n_q}, (\ell-1)\delta n_q)$, 
$\widehat{Y}= q(hX+N)$, $N\sim \mathcal{N}(0,1)$, $\mathcal{P}_{\mathbf{x}}(P_T)$ is the set of distributions on  $\{x_1,x_2,\cdots,x_{\gamma+1}\}$ such that $\mathbb{E}(X^2)\leq P_T$,
and $q(y)=k$ if $y\in [t_{k},t_{k+1}], k\in \{1,\cdots,\gamma-1]$ and $q(y)=0$ if $y>t_{\gamma-1}$ or $y<t_{1}$. Furthermore we have defined $\gamma'\triangleq \min(\ell^{n_q}-1, (\ell-1)\delta n_q)$,
$\widehat{Y'}= q'(hX+N)$, and $q'(y)= k$ if $y\in [t_{k},t_{k+1}], k\in \{1,\cdots,\gamma'-1]$ and $q(y)=0$ if $y<t_{1}$ and $q(y)=\gamma'$ if $y>t_{\gamma'-1}$.
\end{Theorem}
\color{black}
We make the following observations regarding the achievable regions in Theorems \ref{th:4} and \ref{th:5}: 
\\1) Equations \eqref{eq:th:2} and \eqref{eq:th:3} imply that the capacity expressions for odd and even $\delta$ are different.  The reason is that even-degree polynomials yield the same output sign as their input converges to $\pm \infty$, for odd-degree polynomials the output signs are different  as their input converges to $\pm \infty$. This is in contrast with Theorem \ref{th:4}, where the capacity expression is the same for even and odd values of $\delta$. The reason is that the for absolute values the output sign is positive on both sides.
\\2) The achievable region of Theorem \ref{th:5} strictly subsumes Theorem \ref{th:4} for equal $\gamma$ since envelope detectors generate absolute value functions which force a symmetric structure on the Voronoi regions of $q(\cdot)$. This manifests in the fully-symmetric condition $\mathbf{t}\in \mathcal{T}_{\gamma_\ell}$ in Theorem \ref{th:4}; whereas for polynomial functions, no such symmetry is required. 
\\3) One potential approach to improve the capacity in Theorem \ref{th:4} is to augment the envelope detectors by linearly combining their output with the original signal. That is, to generate operators of the form $f(y)=|y-a|+by, a,b\in \mathbb{R}$ instead of $f(y)=|y-b|, b\in \mathbb{R}$. This removes the fully-symmetric condition $\mathbf{t}\in \mathcal{T}_{\gamma_\ell}$. However, such linear combinations are challenging to implement using analog circuits due to timing issues in synchronizing the output of the envelope detector with the original signal. 

\section{A Hybrid Beamforming Architecture with One-bit ADCs}
\label{sec:hyb}
In the previous sections, we have investigated the channel capacity for SISO systems equipped with different collections of implementable analog functions. 
In this section, we consider MIMO systems with one-bit ADCs. We provide a quantization setup, using envelope detectors, which accommodates QAM modulation, and derive an inner bounds to the system capacity. In the next section, we numerically evaluate the resulting capacity and provide comparisons with prior works. 
\subsection{Quantizer Construction} 
We assume that $n_q>n_r$. As mentioned in Section \ref{sec:form} a quantizer is characterized by its analog processing functions $f_j(\cdot), j\in [n_q]$ and ADC thresholds\footnote{We consider one-bit ADCs in this section, hence, each ADC has only one threshold value.} $t^{n_q}\in \mathbb{R}^{n_q}$. Let us fix a threshold step parameter $\zeta>0$. We take the analog processing functions as follows:
\begin{align*}
    f_j({y}^{n_r})=
    \begin{cases}
       {y}_j \quad &\text{ if } j\leq n_r,\\
        |{y}_{\bar{j}}|& \text{ if } j>n_r,
    \end{cases} \qquad 
\end{align*}
where $\bar{j}$ is the modulo $n_r$ residual of $j$. We take the threshold values as follows:
\begin{align}
    t_j=
    \begin{cases}
        0 \quad &\text{ if } j\leq n_r,\\
        \floor{\frac{j}{n_q}}\zeta& \text{ if } j>n_r.
    \end{cases} \qquad 
    \label{eq:zeta}
\end{align}
This choice is clarified through the following example. 
 \begin{figure}[t]
\centering 
\includegraphics[width=0.25\textwidth]{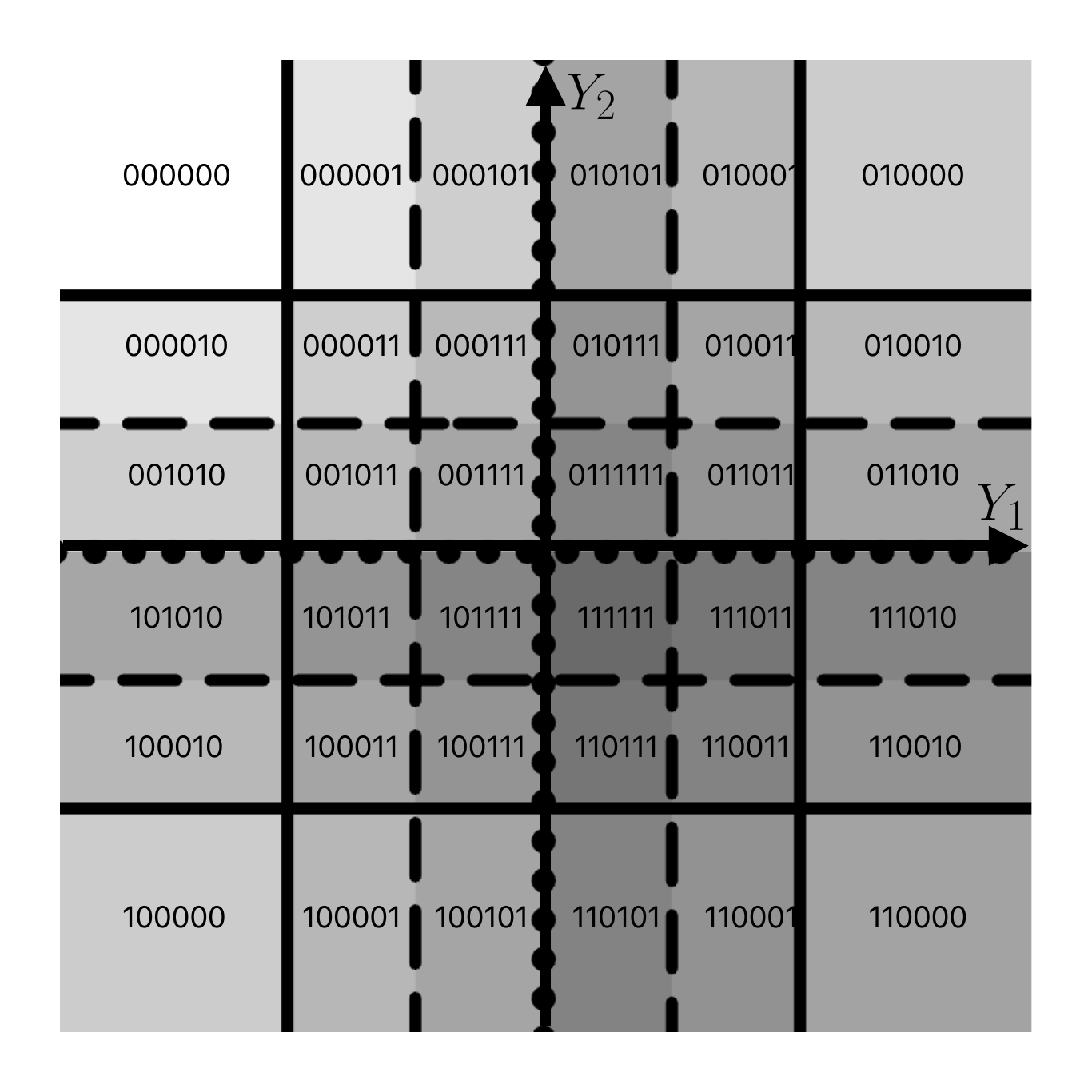}
\caption{Quantizer outputs and Voronoi regions in Example \ref{ex:4}.}
\vspace{-.2in}
\label{fig:quant}
\end{figure}
\begin{Example}[\textbf{Choice of Thresholds}]
\label{ex:4}
    Let $n_r=2$, $n_q=6$, and $\zeta=1$. Then, for the construction described above, the six one-bit ADC operations are as follows:
    \begin{align*}
       & q_1:{y}_1\lessgtr 0, \quad  q_2:{y}_2\lessgtr 0, \quad 
        q_3:|{y}_1|\lessgtr 1, \quad 
        \\&q_4:|{y}_2|\lessgtr 1, \quad 
        q_5:|{y}_1|\lessgtr 2, \quad 
        q_6:|{y}_2|\lessgtr 2.
    \end{align*}
    The quantizer outputs are shown in Figure \ref{fig:quant}. Note that this resembles a 16-QAM modulation.  
\end{Example}
\subsection{Achievable Rates at High SNR}
\color{black}
As argued in Section \ref{sec:env}, the high SNR capacity is equal to the maximum number of Voronoi regions which can be generated given the number of ADCs $n_q$ with $\ell$ output levels, and the set of implementable analog functions. The following theorem provides upper and lower bounds on the high SNR channel capacity of a MIMO system equipped with $n_q$ one-bit ADCs. The lower-bound is based on MIMO systems equipped with envelope detectors, whereas the upper-bound holds generally regardless of analog processing restrictions.

\begin{figure*}[!ht]
\centering
  \begin{subfigure}[b]{0.3\textwidth}
  \includegraphics[width=\linewidth]{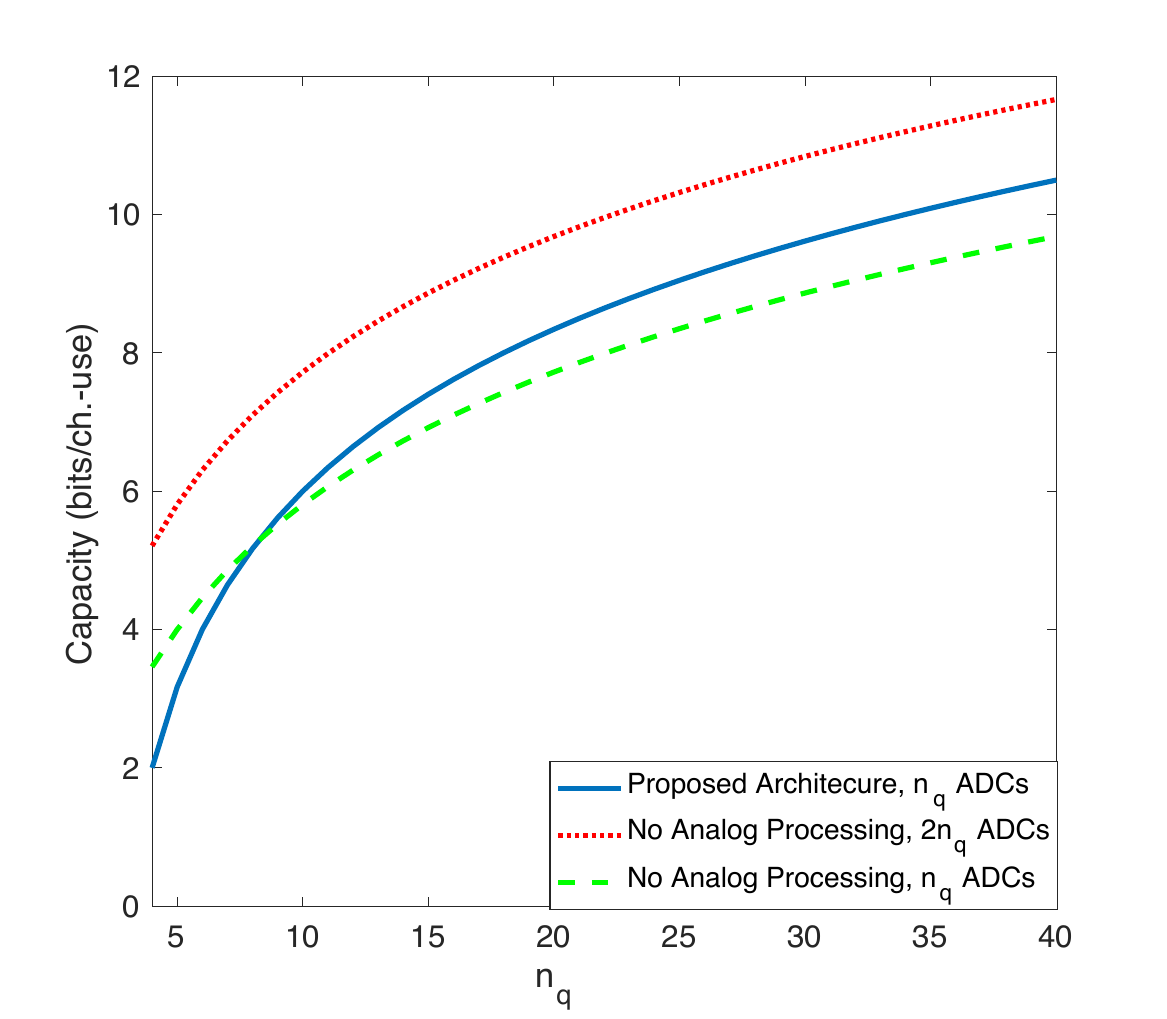}
  \vspace{-.25in}
    \caption{}
  \end{subfigure}
  \hfill
  \begin{subfigure}[b]{0.33\textwidth}
    \includegraphics[width=\linewidth]{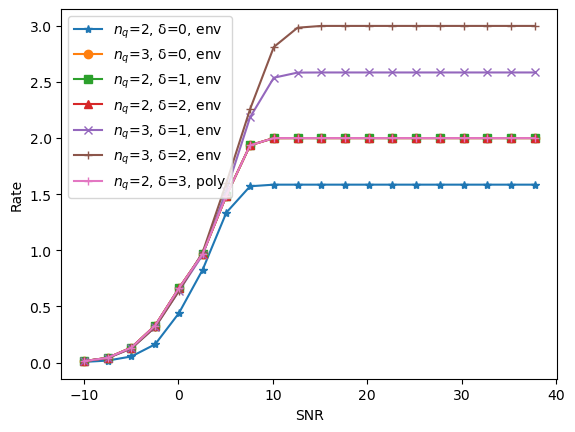}
    \vspace{-.25in}
    \caption{}
  \end{subfigure}
  \hfill
  \begin{subfigure}[b]{0.33\textwidth}
    \includegraphics[width=\linewidth]{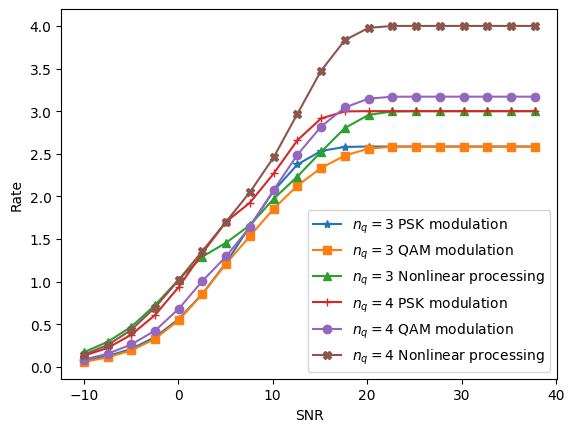}
    \vspace{-.25in}
    \caption{}
  \end{subfigure}
    \caption{\textcolor{black}{(a) Channel capacity for i) the proposed architecture and $n_q$ one-bit ADCs (lower-bound in Theorem \ref{th:6}), ii) no analog processing and $2n_q$ one-bit ADCs (upper-bound in Theorem \ref{th:6})  and
iii) no analog processing and $n_q$ one-bit ADCs (~\cite{khalili2021mimo}).
(b) Achievable rates for various $(n_q,\delta)$ with $\ell=2$ using envelope detectors; $\delta=0$ represents the baseline without analog processing. 
(c) Achievable rates for the hybrid beamforming architecture of Section \ref{sec:hyb}.}
}
\label{fig:simulations}
\vspace{-0.2in}
\end{figure*}

\begin{Theorem}
\label{th:6}
Given $P_T,P_{ADC}>0$ and $n_r,n_q \in \mathbb{N}$, let $C_{env}(P_T,P_{ADC},n_r,n_q)$ denote the channel capacity of the MIMO system. Then, 
    \begin{align}
    \label{eq:hyb}
    &   n_r\left(1+\log{(\ell-1)}+\log\left(\frac{n_q-n_r}{n_r}\right)\right)
 \!\!\leq\!\! \lim_{P_T\to \infty} 
   C_{env}(P_T,P_{ADC},n_r,n_q) 
   \\& \nonumber \qquad \qquad\qquad\leq \log\sum_{k = 0}^{n_r} {2(\ell-1)n_q \choose k},
    \end{align}
    where $n_q= \floor{\frac{P_{ADC}}{2\alpha}}$.
\end{Theorem}
\color{black}

The lower bound in Equation \eqref{eq:hyb} is achieved by the quantizer introduced in the previous section. To see this, note that by construction,  the quantizer partitions each axis into $2(\ell-1)\left(\frac{n_q-n_r}{n_r}\right)$ intervals, and each resulting quantization region is mapped to a unique quantizer output (e.g., Figure \ref{fig:quant}). So,  the total number of unique quantizer outputs is $|\mathcal{C}|= (2(\ell-1)(\frac{n_q-n_r}{n_r}))^{n_r}$. The result follows by noting that the communication rate is $\log{|\mathcal{C}|}$. 
The upper bound follows by counting the number of partition regions generated by $2(\ell-1)n_q$ hyperplanes in general position in the $n_r$-dimensional Euclidean space (e.g. \cite{khalili2021mimo}). Figure \ref{fig:simulations}(a) provides numerical simulations of the i) upper bounds and ii) lower bounds in Equation \eqref{eq:hyb} and  iii) the high SNR channel capacity under hybrid beamforming without analog processing derived in \cite{khalili2021mimo}, where we have taken $n_r=3$. It can be observed that the proposed architecture outperforms the one in \cite{khalili2021mimo} if the number of one-bit ADCs is larger than $n_q=8$. 



\section{Numerical Analysis of Channel Capacity}
\label{sec:num}
We compute an inner-bound to the capacity expression of the SISO system in Theorem \ref{th:4}, for various SNR values and as a function of the number of ADCs $n_q$ and the number of concatenated envelope detectors $\delta$. 
\textcolor{black}{To this end, we conduct a brute-force search as described in the following. Recall that the optimization is over three sets of parameters, namely, the input points, $\mathbf{x}$, threshold vector $\mathbf{t}$, and probability distribution $P_X$. 
To search for the value of $\textbf{x}$ and $\textbf{t}$, we discretize the interval $[-3,3]$ using a grid  with step-size $0.1$. We search over all symmetric vectors  $\textbf{t}$ over the grid, where symmetry is with respect to the origin. We place each input point ${x}_i$ in the center of the corresponding quantization interval $[t_i,t_{i+1}]$. 
Furthermore, we optimize the distribution $P_X$ over the resulting discrete space via the Blahut-Arimoto algorithm. That is,  for each configuration, we find optimal input distribution, using the modified Blahut-Arimoto algorithm for discrete memoryless channels with input cost constraints given in \cite{kobayashi2018joint}.}
Fig.  \ref{fig:simulations}(b) shows the achievable rates for SNRs in the range of -10 to 40 dB for various values of $(n_q,\delta)$. 
Observer that for $n_q=\delta=2$, the inner bound in Figure \ref{fig:simulations}(b) at SNRs higher than 15dB surpasses the high SNR capacity of systems without nonlinear analog processing with $n_q=2$ and $\delta=0$, and the high SNR capacity of the former is more than $25\%$ greater than that of the latter.  
Note that the achievable rate for quadratic operator architectures, for appropriately chosen parameters, would completely match the ones shown in Figure \ref{fig:simulations}(b). For instance, for a polynomial operators with $\delta=2$, if the input to the ADC is $(Y-a)^2$ and the threshold is $b$, then the setup is equivalent to feeding $|Y-a|$ to an ADC with threshold $\sqrt{b}$. Consequently, we do not plot the resulting achievable rates using polynomial operators with $\delta=2$,  since they would match the ones in Figure \ref{fig:simulations}(b) for envelop detectors with $\delta=1$. 

 \begin{figure*}[!ht]
\centering 
\includegraphics[width=0.9\textwidth]{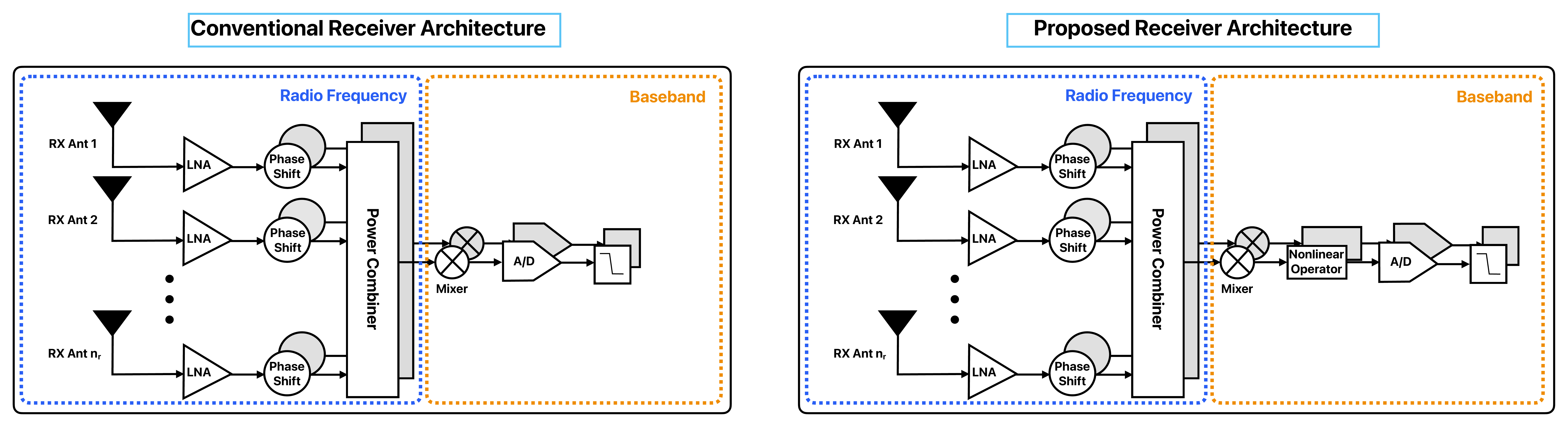}
\caption{\textcolor{black}{A hybrid beamforming receiver using (left) conventional architectures and (right) proposed architecture.}}
\label{receiver}
\end{figure*}
\begin{figure*}
    \centering

\begin{minipage}{\textwidth}
\color{black}
        \centering
        \textbf{Power Breakdown for Analog Beamforming with 4 Antennas}
        
        \scriptsize 
        \setlength{\tabcolsep}{2pt} 
        \begin{tabular}{|c|c|c|c|c|c|c|c|c|c|c|}
            \hline
            \textbf{Receiver} & 
            \textbf{LNA} & 
            \textbf{Phase Shifters} & 
            \textbf{Combiners} & 
            \textbf{Mixers} & 
            \textbf{NL Operator} & 
            \textbf{ADC\cite{4Gadc}$^\ddagger$} & 
            \textbf{Total (mW)} &
            \textbf{Rate} & 
            \textbf{Rate/Power} & 
            \textbf{\% Improvement}\\\hline
            
            Linear Operator &
            \textit{$4\times10$} &
            \textit{$4\times11.5$} &
            \textit{$1\times19.5$} &
            \textit{$1\times14.2$} &
            \textit{$0$} & 
            \textit{$1\times2\times70fJ\times2^{8}\times10^9\times10^3$} & 
            \textit{$155.54$} & 
            \textit{$8$} & 
            \textit{$0.05143$} &
            \textit{--} \\ \hline
            
            1 Env. Detector &
            \textit{$4\times10$} &
            \textit{$4\times11.5$} &
            \textit{$1\times19.5$} &
            \textit{$1\times14.2$} &
            \textit{$1\times3$} & 
            \textit{$8\times2\times70fJ\times2^{3}\times10^9\times10^3$} & 
            \textit{$131.66$}& 
            \textit{$6.807$}& 
            \textit{0.05170}&
            \textit{0.53\%}\\ \hline

            3 Env. Detectors &
            \textit{$4\times10$} &
            \textit{$4\times11.5$} &
            \textit{$1\times19.5$} &
            \textit{$1\times14.2$} &
            \textit{$1\times9$} & 
            \textit{$8\times2\times70fJ\times2^{3}\times10^9\times10^3$}& 
            \textit{$137.66$}& 
            \textit{$8.807$}& 
            \textit{0.06398}&
            \textit{24.39\%}\\\hline
            
             6 Env. Detectors & 
             \textit{$4\times10$} & 
             \textit{$4\times11.5$} & 
             \textit{$1\times19.5$} & 
             \textit{$1\times14.2$} & 
             \textit{$1\times18$}& 
             \textit{$8\times2\times70fJ\times2^{3}\times10^9\times10^3$}& 
             \textit{$146.66$} & 
             \textit{$11.807$} & 
             \textit{0.08051} &
             \textbf{56.53\%}\\\hline 

            $2^{nd}$ degree Poly. &
            \textit{$4\times10$} &
            \textit{$4\times11.5$} &
            \textit{$1\times19.5$} &
            \textit{$1\times14.2$} &
            \textit{$1\times5$} & 
            \textit{$8\times2\times70fJ\times2^{3}\times10^9\times10^3$}& 
            \textit{$133.66$}& 
            \textit{$6.807$}& 
            \textit{0.05093}&
            --\\ \hline

            $4^{th}$ degree Poly. &
            \textit{$4\times10$} &
            \textit{$4\times11.5$} &
            \textit{$1\times19.5$} &
            \textit{$1\times14.2$} &
            \textit{$1\times10$} & \textit{$8\times2\times70fJ\times2^{3}\times10^9\times10^3$} & 
            \textit{$138.66$}& 
            \textit{$7.807$}& 
            \textit{0.05631}&
            \textit{9.47\%}\\ \hline
        \end{tabular}
    \end{minipage}
    \vspace{0.1in} 
    
    \begin{minipage}{\textwidth}
    \color{black}
        \centering
        \textbf{Power Breakdown for Analog Beamforming with 16 Antennas}
        
        \scriptsize 
        \setlength{\tabcolsep}{2pt} 
        \begin{tabular}{|c|c|c|c|c|c|c|c|c|c|c|}
            \hline
            \textbf{Receiver} & 
            \textbf{LNA} & 
            \textbf{Phase Shifters} & 
            \textbf{Combiners} & 
            \textbf{Mixers} & 
            \textbf{NL Operator} & 
            \textbf{ADC\cite{4Gadc}$^\ddagger$} & 
            \textbf{Total (mW)} &
            \textbf{Rate} & 
            \textbf{Rate/Power} & 
            \textbf{\% Improvement}\\\hline
            
            Linear Operator &
            \textit{$16\times10$} &
            \textit{$16\times11.5$} &
            \textit{$1\times19.5$} &
            \textit{$1\times14.2$} &
            \textit{$0$} & 
            \textit{$1\times2\times70fJ\times2^{8}\times10^9\times10^3$} & 
            \textit{$413.54$} & 
            \textit{$8$} & 
            \textit{$0.01935$} &
            --\\ \hline
            
            1 Env. Detector &
            \textit{$16\times10$} &
            \textit{$16\times11.5$} &
            \textit{$1\times19.5$} &
            \textit{$1\times14.2$} &
            \textit{$1\times3$} & 
            \textit{$16\times2\times70fJ\times2^{3}\times10^9\times10^3$} & 
            \textit{$398.62$} & 
            \textit{$7.807$}& 
            \textit{0.01958}&
            \textit{1.24\%}\\ \hline

            3 Env. Detectors &
            \textit{$16\times10$} &
            \textit{$16\times11.5$} &
            \textit{$1\times19.5$} &
            \textit{$1\times14.2$} &
            \textit{$1\times9$} & 
            \textit{$16\times2\times70fJ\times2^{3}\times10^9\times10^3$} & 
            \textit{$404.62$}& 
            \textit{$9.807$}& 
            \textit{0.02424}&
            \textit{25.29\%}\\\hline
            
             6 Env. Detectors & 
             \textit{$16\times10$} & 
             \textit{$16\times11.5$} & 
             \textit{$1\times19.5$} & 
             \textit{$1\times14.2$} & 
             \textit{$1\times18$}& 
             \textit{$16\times2\times70fJ\times2^{3}\times10^9\times10^3$} & 
             \textit{$413.62$}& 
             \textit{$12.807$}& 
             \textit{0.03096}&
             \textbf{60.06\%}\\\hline 

            $2^{nd}$ degree Poly. &
            \textit{$16\times10$} &
            \textit{$16\times11.5$} &
            \textit{$1\times19.5$} &
            \textit{$1\times14.2$} &
            \textit{$1\times5$} & 
            \textit{$16\times2\times70fJ\times2^{3}\times10^9\times10^3$} & 
            \textit{$400.62$}& 
            \textit{$7.807$}& 
            \textit{0.01949}&
            \textit{0.74\%}\\ \hline

            $4^{th}$ degree Poly. &
            \textit{$16\times10$} &
            \textit{$16\times11.5$} &
            \textit{$1\times19.5$} &
            \textit{$1\times14.2$} &
            \textit{$1\times10$} & \textit{$16\times2\times70fJ\times2^{3}\times10^9\times10^3$} & 
            \textit{$405.62$}& 
            \textit{$8.807$}& 
            \textit{0.02171}&
            \textit{12.24\%}\\ \hline
        \end{tabular}
    \end{minipage}
    \vspace{0.1in} 
    
    \begin{minipage}{\textwidth}
    \color{black}
        \centering
        \textbf{Power Breakdown for Hybrid Beamforming with 4 Antennas}
        
        \scriptsize 
        \setlength{\tabcolsep}{2pt} 
        \begin{tabular}{|c|c|c|c|c|c|c|c|c|c|c|}
            \hline
            \textbf{Receiver} & 
            \textbf{LNA} & 
            \textbf{Phase Shifters} & 
            \textbf{Combiners} & 
            \textbf{Mixers} & 
            \textbf{NL Operator} & 
            \textbf{ADC\cite{4Gadc}$^\ddagger$} & 
            \textbf{Total (mW)} &
            \textbf{Rate} & 
            \textbf{Rate/Power} & 
            \textbf{\% Improvement}\\\hline
            
            Linear Operator &
            \textit{$4\times10$} &
            \textit{$8\times11.5$} &
            \textit{$2\times19.5$} &
            \textit{$2\times14.2$} &
            \textit{$0$} & 
            \textit{$2\times2\times70fJ\times2^{8}\times10^9\times10^3$} & 
            \textit{$271.08$} & 
            \textit{$16$} & 
            \textit{$0.05902$} &
            \textit{--} \\ \hline
            
            1 Env. Detector &
            \textit{$4\times10$} &
            \textit{$8\times11.5$} &
            \textit{$2\times19.5$} &
            \textit{$2\times14.2$} &
            \textit{$2\times3$} & \textit{$8\times2\times70fJ\times2^{3}\times10^9\times10^3$} & 
            \textit{$214.36$} & 
            \textit{$11.6147$} & 
            \textit{0.05418} &
            \textit{--} \\ \hline

            3 Env. Detectors &
            \textit{$4\times10$} &
            \textit{$8\times11.5$} &
            \textit{$2\times19.5$} &
            \textit{$2\times14.2$} &
            \textit{$2\times9$} & \textit{$8\times2\times70fJ\times2^{3}\times10^9\times10^3$} & 
            \textit{$226.36$}& 
            \textit{$15.6147$} & 
            \textit{0.06898}&
            \textit{16.87\%}\\ \hline

            6 Env. Detectors &
            \textit{$4\times10$} &
            \textit{$8\times11.5$} &
            \textit{$2\times19.5$} &
            \textit{$2\times14.2$} &
            \textit{$2\times18$} & \textit{$8\times2\times70fJ\times2^{3}\times10^9\times10^3$} & 
            \textit{$244.36$}& 
            \textit{$21.6147$} & 
            \textit{0.08845}&
            \textbf{49.86\%}\\ \hline

            $2^{nd}$ degree Poly. &
            \textit{$4\times10$} &
            \textit{$8\times11.5$} &
            \textit{$2\times19.5$} &
            \textit{$2\times14.2$} &
            \textit{$2\times5$} & \textit{$8\times2\times70fJ\times2^{3}\times10^9\times10^3$} & 
            \textit{$218.36$}& 
            \textit{$11.6147$} & 
            \textit{0.05319}&
            \textit{--}\\ \hline

            $4^{th}$ degree Poly. &
            \textit{$4\times10$} &
            \textit{$8\times11.5$} &
            \textit{$2\times19.5$} &
            \textit{$2\times14.2$} &
            \textit{$2\times10$} & \textit{$8\times2\times70fJ\times2^{3}\times10^9\times10^3$} & 
            \textit{$228.36$}& 
            \textit{$13.6147$} & 
            \textit{0.05962}&
            \textit{1.01\%}\\ \hline
        \end{tabular}
    \end{minipage}
    \vspace{0.1in} 
    
    \begin{minipage}{\textwidth}
    \color{black}
        \centering
        \textbf{Power Breakdown for Hybrid Beamforming with 16 Antennas}
        
        \scriptsize 
        \setlength{\tabcolsep}{2pt} 
        \begin{tabular}{|c|c|c|c|c|c|c|c|c|c|c|}
            \hline
            \textbf{Receiver} & 
            \textbf{LNA} & 
            \textbf{Phase Shifters} & 
            \textbf{Combiners} & 
            \textbf{Mixers} & 
            \textbf{NL Operator} & 
            \textbf{ADC\cite{4Gadc}$^\ddagger$} & 
            \textbf{Total (mW)} &
            \textbf{Rate} & 
            \textbf{Rate/Power} & 
            \textbf{\% Improvement}\\\hline
            
            Linear Operator &
            \textit{$16\times10$} &
            \textit{$32\times11.5$} &
            \textit{$2\times19.5$} &
            \textit{$2\times14.2$} &
            \textit{$0$} & 
            \textit{$2\times2\times70fJ\times2^{8}\times10^9\times10^3$} & 
            \textit{$667.08$} & 
            \textit{$16$} & 
            \textit{$0.02399$} &
            \textit{--} \\ \hline
            
            1 Env. Detector &
            \textit{$16\times10$} &
            \textit{$32\times11.5$} &
            \textit{$2\times19.5$} &
            \textit{$2\times14.2$} &
            \textit{$2\times3$} & \textit{$32\times2\times70fJ\times2^{3}\times10^9\times10^3$} & 
            \textit{$637.24$}& 
            \textit{$15.6147$} & 
            \textit{0.02450}&
            \textit{2.16\%}\\ \hline

            3 Env. Detectors &
            \textit{$16\times10$} &
            \textit{$32\times11.5$} &
            \textit{$2\times19.5$} &
            \textit{$2\times14.2$} &
            \textit{$2\times9$} & \textit{$32\times2\times70fJ\times2^{3}\times10^9\times10^3$} & 
            \textit{$649.24$}& 
            \textit{$19.6147$} & 
            \textit{0.03021}&
            \textit{25.96\%}\\ \hline

            6 Env. Detectors &
            \textit{$16\times10$} &
            \textit{$32\times11.5$} &
            \textit{$2\times19.5$} &
            \textit{$2\times14.2$} &
            \textit{$2\times18$} & \textit{$32\times2\times70fJ\times2^{3}\times10^9\times10^3$} & 
            \textit{$667.24$}& 
            \textit{$25.6147$} & 
            \textit{0.03839}&
            \textbf{60.05\%}\\ \hline

            $2^{nd}$ degree Poly. &
            \textit{$16\times10$} &
            \textit{$32\times11.5$} &
            \textit{$2\times19.5$} &
            \textit{$2\times14.2$} &
            \textit{$2\times5$} & \textit{$32\times2\times70fJ\times2^{3}\times10^9\times10^3$} & 
            \textit{$641.24$}& 
            \textit{$15.6147$} & 
            \textit{0.02435}&
            \textit{1.52\%}\\ \hline

            $4^{th}$ degree Poly. &
            \textit{$16\times10$} &
            \textit{$32\times11.5$} &
            \textit{$2\times19.5$} &
            \textit{$2\times14.2$} &
            \textit{$2\times10$} & \textit{$32\times2\times70fJ\times2^{3}\times10^9\times10^3$} & 
            \textit{$651.24$}& 
            \textit{$17.6147$} & 
            \textit{0.02705}&
            \textit{12.77\%}\\ \hline
        \end{tabular}
    \end{minipage}
    \vspace{0.1in} 

    $^\dagger$ {For each RX, a four bit binary weighted current adder is implemented with 1.6mW power consumption}\\
    $^\ddagger$ {Both power consumption of 3-bit-ADC and eight-bit ADC are normalized from the reference}\\
    
    \vspace{0in} 

    \caption{The breakdown of power consumption for the receiver in a conventional architecture and the proposed architecture.}
    \label{fig:two_tables}
\end{figure*}
\textcolor{black}{
We further evaluate the achievable  rates using the beamforming architecture described in Section \ref{fig:simulations}(c). We numerically simulate the communication system for $n_r=2$, $\ell=2$, and $n_q\in \{3,4\}$ and find an estimate of the achievable rates empirically ( Figure \ref{fig:simulations}(c)). To compute the rate, we have optimized the parameter $\zeta$ in Equation \eqref{eq:zeta} via brute-force search over the interval [0,1] with step-size 0.01. We have simulated the channel by generating 150000 independent and identically distributed samples of noise vectors and input messages, and have empirically  estimated the transition probability of the discrete channel resulting from the quantization process. We have found the optimal probability distribution $P_X$ using the Blahut-Arimoto algorithm.  
It can be observed in Figure \ref{fig:simulations}(c) that for both $n_q=3$ and $n_q=4$, the high SNR rate is larger than the $n_r\left(1+\log{\left(\frac{n_q-n_r}{n_r}\right)}\right)$ lower-bound given in Theorem \ref{th:6}, and it significantly larger than the linear processing baselines.  }
\color{black}

\section{Circuit Design and Power Analysis}
\label{sec:cir}

\begin{figure*}[!ht]
\centering
  \begin{subfigure}[b]{0.3\textwidth}
  \includegraphics[width=\linewidth]{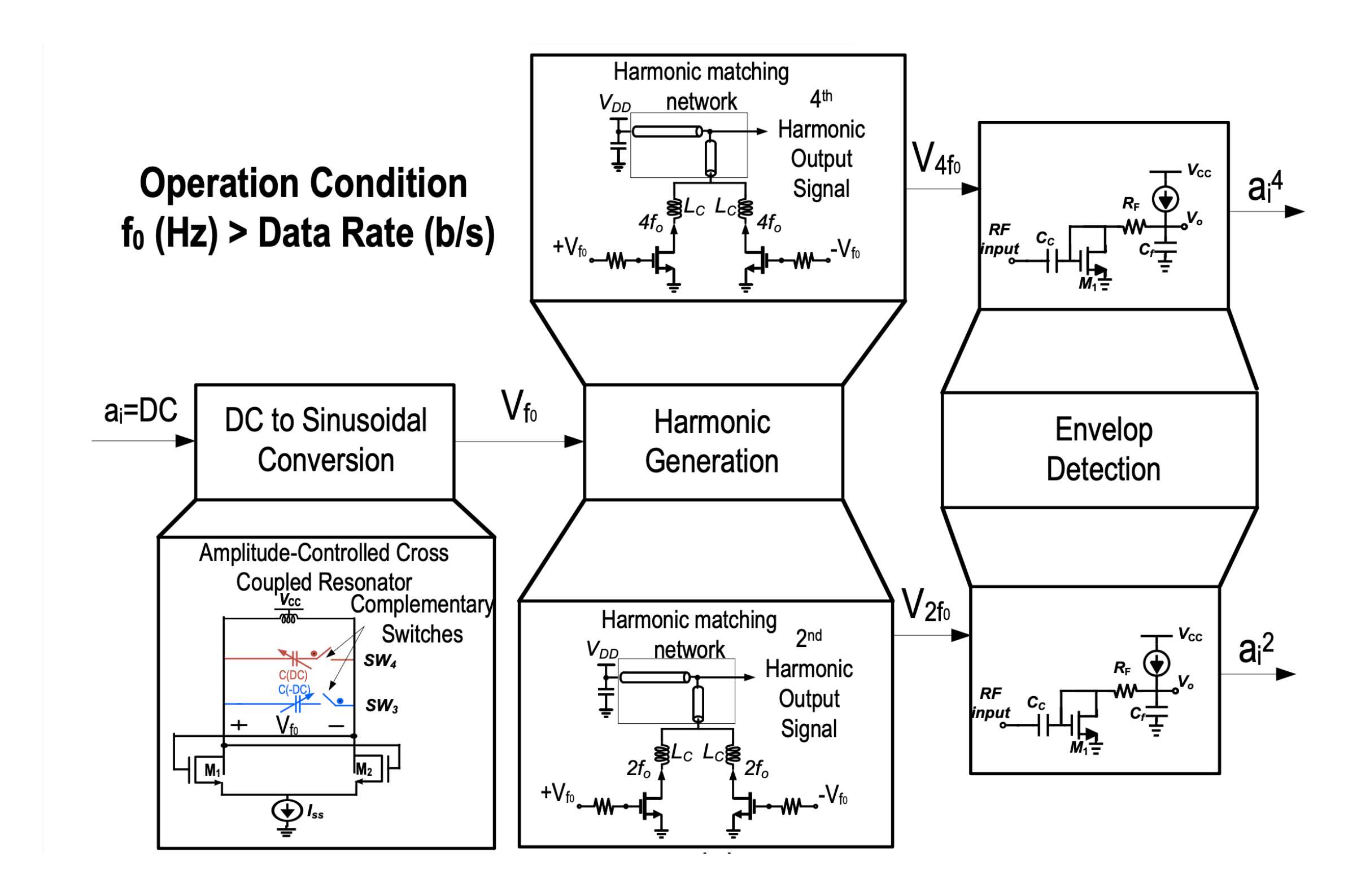}
    \caption{}
  \end{subfigure}
  \hfill
  \begin{subfigure}[b]{0.33\textwidth}
    \includegraphics[width=\linewidth]{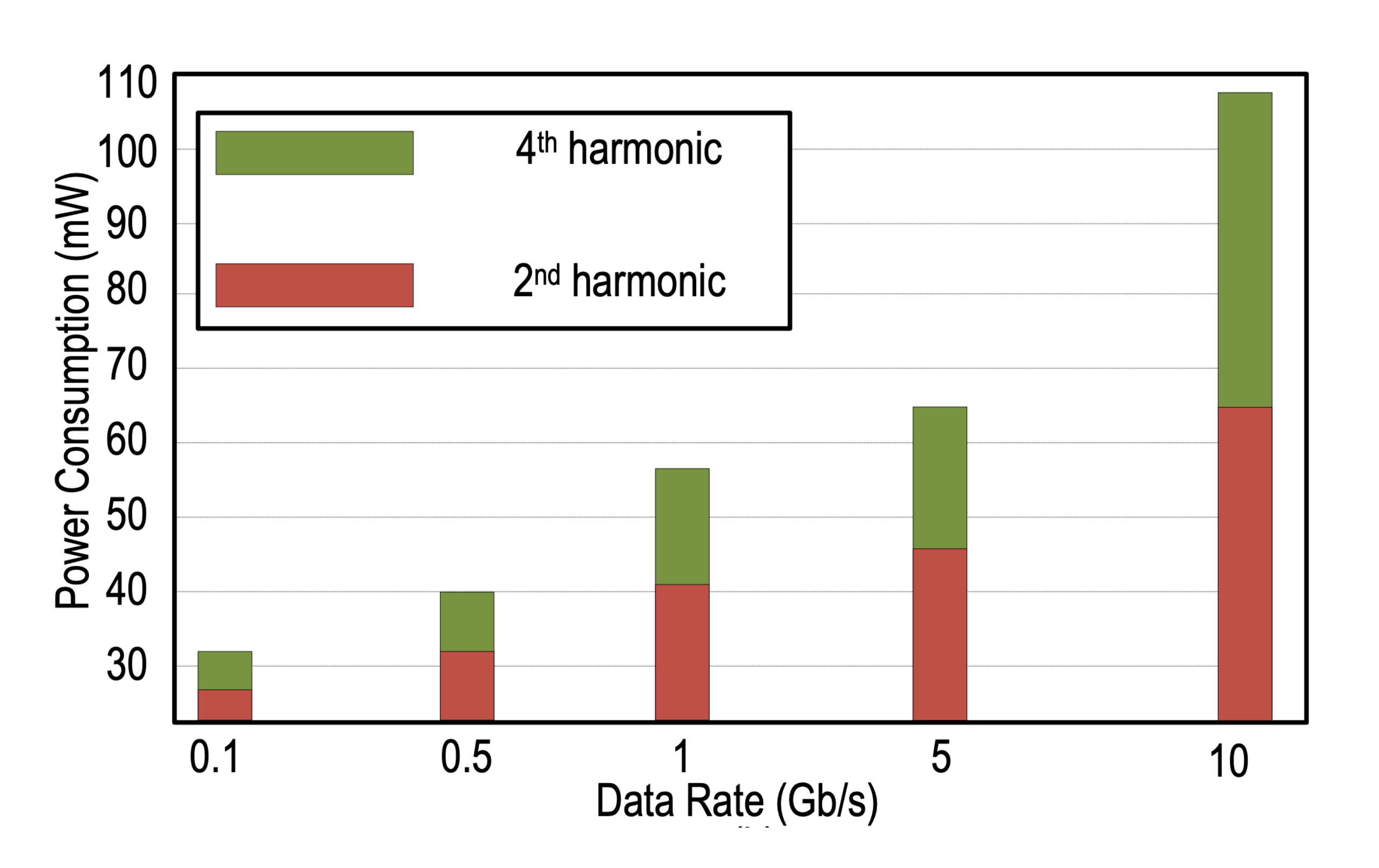}
    \caption{}
  \end{subfigure}
  \hfill
  \begin{subfigure}[b]{0.33\textwidth}
    \raisebox{0.7\height}{\includegraphics[width=\linewidth]{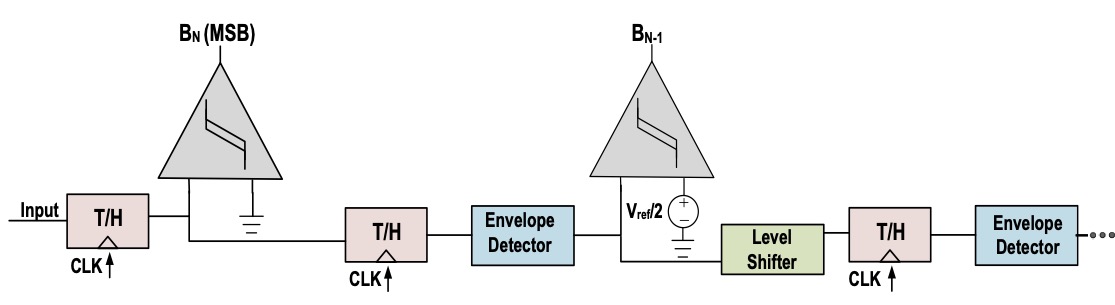}}
    \caption{}
  \end{subfigure}
    \caption{(a) Generation of fourth and second order polynomials. (b) Power consumption for generation of equal voltage amplitude (0 dBm power) at the second and fourth harmonics. (c) The cascade of circuits emulating absolute value operators.}
\label{fig:poly}
\end{figure*}
\begin{figure}[t]
\centering 
\includegraphics[width=0.4\textwidth]{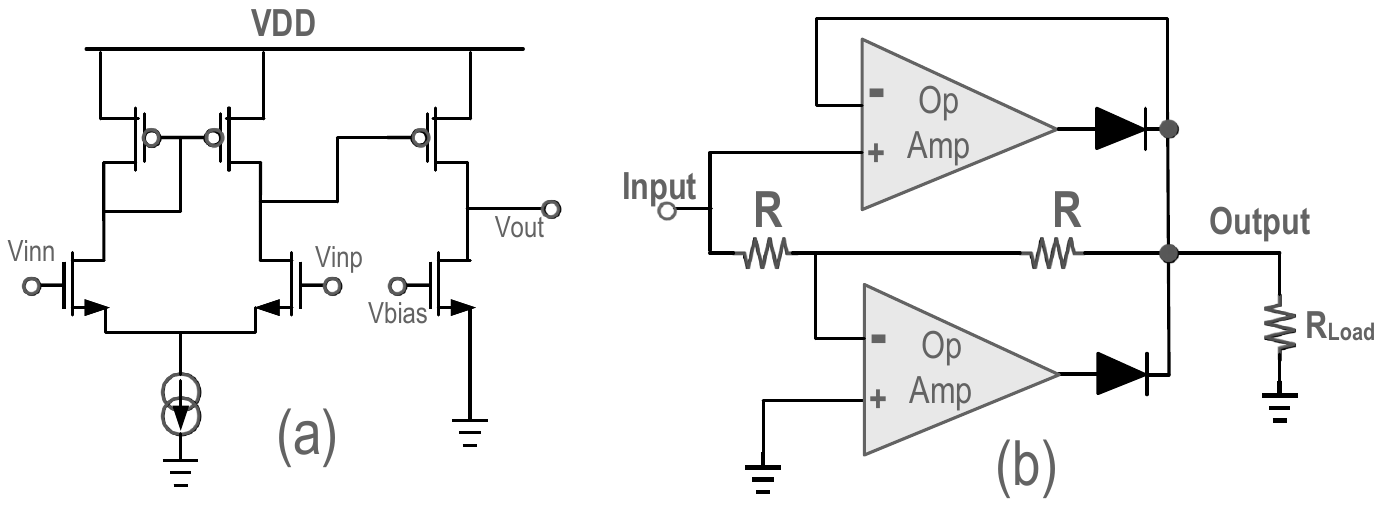}
\caption{(a) Differential to single-ended op-amp in the envelope detector (b) Circuit diagram of the envelope detector}
\label{Xuyang1}
\vspace{-0.2in}
\end{figure}

\color{black}
Figure \ref{receiver} shows the architecture of a conventional hybrid beamforming receiver (left) and the proposed receiver (right). In the Figure \ref{fig:two_tables}, the itemized and total power consumption is compared between the two architectures in a 22 nm FDSOI CMOS technology. For the power values in Figure \ref{fig:two_tables}, we assume a mm-wave carrier frequency of 60 GHz, RF bandwidth of 1 GHz, {and 16 channels}. The reported power consumption for RF elements are based on the measurement results of the chip reported in \cite{xuyang} and the power consumption of nonlinear operators and ADCs are based on post-layout simulation results and analytical values in \cite{Razavi_paper}, respectively. To simplify the rate analysis we assume communication takes place at the high SNR regime (Theorems \ref{th:4}, \ref{th:5}, and \ref{th:6}). 

\subsection{Analysis Based on 22 nm FDSOI CMOS Technology}
\textbf{Analog Beamforming:} 
The receive antennas are connected to a low noise amplifier (LNA),  
and the amplified signals are combined using phase-shifters. A mixer is used to produce the baseband signal. For the linear architecture, we consider a single eight-bit ADC. For the nonlinear architectures, we have used the indicated number of non-linear operators, including envelope detectors, and second degree and fourth degree polynomials. We have used a three-bit ADCs per antenna in the sixteen receiver antenna scenario, and two  three-bit ADCs per antenna in the four receiver antenna scenario. The results show significant potential improvements in the rate-power tradeoff by using nonlinear processing prior to quantization. For instance, in the sixteen receiver antenna scenario, it can be seen that the rate to power ratio achieved using the conventional architecture is 0.0193, whereas the architecture introduced in this work achieves ratio up to 0.03096, hence achieving more than $\%60$ improvement. 
\\\textbf{Hybrid Beamforming:} The signals are fed to $2n_r$ phase shifters and then linearly combined into two analog signals. The two resulting analog signals are passed through two mixers and then through individual eight-bit ADCs. For the bottom two tables, we use the indicated number of non-linear analog circuits and two three-bit ADCs per antenna to digitize the signal. There is more than $\%60$ gains, in terms of rate to power ratio.  
\color{black}
\begin{figure}[t]
\centering 
\begin{subfigure}[t]{0.48\textwidth}
\includegraphics[width=\textwidth]{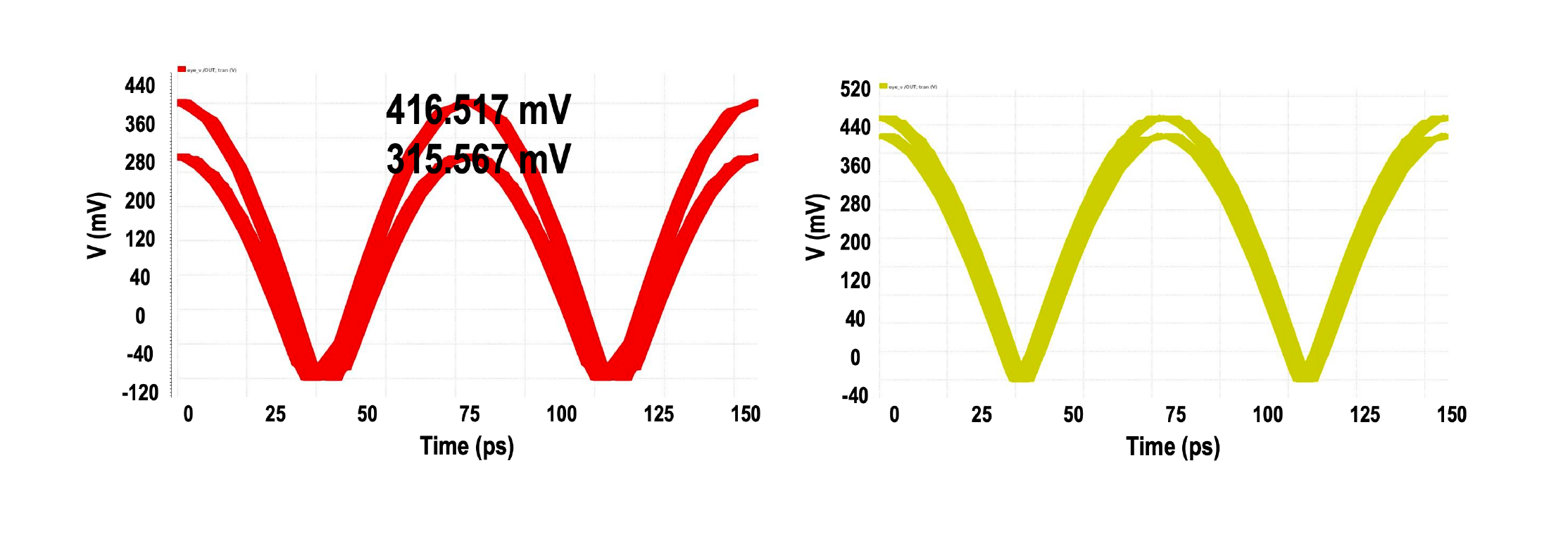}
\vspace{-.34in}
\caption{}
\vspace{-.1in}
\end{subfigure}
\begin{subfigure}[t]{0.48\textwidth}
\includegraphics[width=\textwidth, height=1.6in]{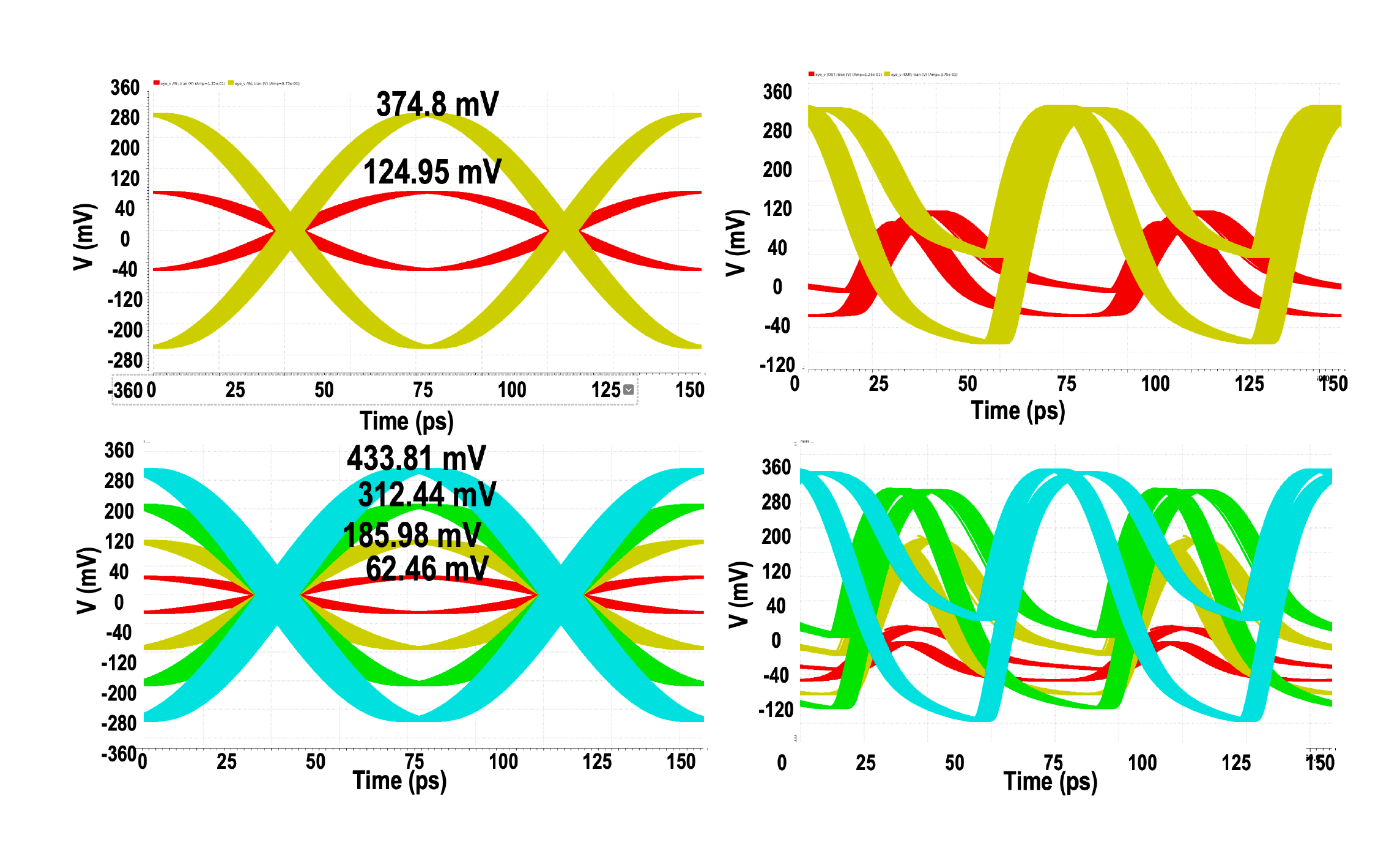}
\vspace{-.3in}
\caption{}
\vspace{-.1in}
\end{subfigure}
\caption{(a) Impact of amplifier gain on the half-cycle rectification amplitude distortion  for  gain of 10 dB (left), and 25 dB (right).
(b) Input waveform eye diagram for (top left) PAM4 and (bottom left) PAM8 modulations compared with that of setups with envelop detectors (right).}
\vspace{-.08in}
\label{Xuyang3}
\vspace{-0.2in}
\end{figure}
\vspace{-0.05in}
\subsection{Analysis Based on 65 nm Bulk CMOS Technology}
We provide two circuit designs for envelope detectors and polynomial operators using the 65 nm Bulk CMOS technology. There are two main advantages to this technology: 1) 65 nm technology has a higher nominal supply voltage than 22 nm. 2) SOI technologies exhibit more threshold voltage variation compared to the bulk CMOS technologies which impacts the operation of baseband circuitry \cite{SOIBULK}.


\noindent \textbf{Polynomial Operators:} 
To explain the proposed construction, consider  $f(x)=x^4+x^2$, where $x$ is the input value.  Fig. \ref{fig:poly}(a) shows a circuit to generate $f(x)=x^4+x^2$. In order to generate equal amplitudes at the second and fourth harmonics, the power gain of the transistors generating the fourth harmonic should be larger, leading to an increased power consumption in generating the fourth order term compared to the second order term. Figure \ref{fig:poly}(b) illustrates numerical values for the power consumption of the proposed circuit through simulations. For large bandwidth communications, a higher frequency for the generation of polynomials in Figure \ref{fig:poly}(a) may be needed.   

\label{envelope_circuit}
\noindent\textbf{Envelope Detectors:} The proposed multi-step envelope detector is shown in Figure \ref{fig:poly}(c). 
 The operational amplifiers deployed in the envelope detector are two-stage differential to single-ended amplifiers with gain-bandwidth product (GBW) of 32 GHz in Figure \ref{Xuyang1}(a). This GBW allows to amplify signals up to 10 GHz with a gain above 15 dB, which is critical for the operation of the envelope detector shown in Figure \ref{Xuyang1}(b). The resistors in this circuit establish a trade-off between the bandwidth and the waveform distortion. In other words, the larger resistance value leads to smaller distortion of the flipping negative part at the expense of increasing the resistance associated with the output pole of each envelope detector stage, and subsequently limiting the bandwidth of operation. In our simulations, we assumed 500$\Omega$ resistors. According to the simulation results shown in Figure \ref{Xuyang3}(a), the higher gain of each operational amplifier leads to smaller amplitude distortion at the output, which naturally is achieved at the expense of a smaller bandwidth for the amplifier, thereby leading to a distortion-bandwidth tradeoff.  

Simulation of two concatenated envelope detector circuits is shown in Figure \ref{fig:poly}(c). The DC level shifter is realized by diode-based circuits which consume no DC power and can operate up to 6 GHz. 
The simulated eye-diagram  performance of the two-stage envelope detector is illustrated in Figure \ref{Xuyang3}(b). For high data rates (12 Gb/s), we consider the input to be the absolute value function circuitry to be a sinusoidal waveform due to the slewing effects \cite{slew}. To evaluate our proposed envelope detector, two modulation scenarios of Pulse Amplitude Modulation PAM4 and PAM8 modulation schemes are considered as the input to the envelope detector, as shown in Figure \ref{Xuyang3}(b). The simulation results demonstrate that the amplitude ratios at the output follow the input amplitude ratios even when the dynamic range of input waveform is only between -400 to +400 mV. 

\section{Conclusion}
The application of nonlinear analog operations in MIMO receivers was considered. A receiver architecture consisting of linear analog combiners, implementable nonlinear analog operators, and few-bit threshold ADCs was designed, and
the fundamental information theoretic performance limits of the resulting communication system were investigated. Furthermore, circuit-level simulations,  were provided to show the implementability of the desired nonlinear analog operators with practical power budgets. 

\bibliographystyle{IEEEtran}
\bibliography{References}

\begin{appendices}
    \section{Proof of Proposition \ref{Prop:1}}
    \label{App:Prop1}
Let $n_q= \lfloor\frac{P_{ADC}}{2\alpha}\rfloor$.
For a given quantizer with associated code $\mathcal{C}$, the high SNR achievable rate is $\log{|\mathsf{P}|}= \log{|\mathcal{C}|}$. So, finding the capacity is equivalent to finding the maximum $|\mathcal{C}|$ over all choices of $\mathbf{q}(\cdot)$.
First, let us prove the converse result. Note that $|\mathcal{C}|\leq 2n_q$ since $\mathbf{c}_0=\mathbf{c}_{2n_q}$. The reason is that for the absolute value function $f_{j,1}(\cdot),j\in [n_q]$, we have $\lim_{y\to \infty}f_{j,1}(y)=\lim_{y\to -\infty}f_{j,1}(y)= \infty $. So, 
\begin{align}
\label{eq:lim}
    c_{1,j}=\lim_{y\to -\infty} \mathbbm{1}(f_{j}-t_{j,1}>0)=\lim_{y\to \infty} \mathbbm{1}(f_{j}-t_{j,1}>0)=c_{2n_q,j}.
\end{align}
As a result, $\log{|\mathcal{C}|}\leq 1+\log{n_q}$. Next, we prove achievability.
Let $t_{j,1}=\frac{n_q+1}{2}, j\in [n_q]$ and $f_{j,1}(y)\triangleq |y-j-\frac{n_q+1}{2}|, j\in [n_q]$, so that the roots of $f_{j}(\cdot)-t_{j,1}$ are $j$ and $j-n_q-1$. Then, $(r_1,r_2,\cdots,r_{2n_q})\!=\!(-n_q,-n_q\!+\!1,\cdots,-1,1,2,\cdots, n_q)$ and \begin{align*}
    c_{i,j}=
    \begin{cases}
    1-\mathbbm{1}(j\leq i)\qquad& \text{if }\quad  i\leq n_q,\\
    \mathbbm{1}(n_q-j+1\leq i-n_q) & \text{otherwise}.
    \end{cases}
\end{align*}
For instance, for $n_q=3$, we have $\mathcal{C}=(111,011,001,$ $000,001,011,111)$. Hence, the only repeated codewords are $\mathbf{c}_0$ and $\mathbf{c}_{2n_q}$. As a result, $|\mathcal{C}|=2n_q$, and $\log{|\mathcal{C}|}=1+\log{n_q}$ is achievable. 
\qed

\section{Proof of Theorem \ref{th:1}}
\label{App:th:1}
We provide an outline of the proof. Note that the input alphabet has at most $2n_q+1$ mass points since,
based on the proof of Proposition \ref{Prop:1}, the channel output can take at most $2n_q$ values. Let the quantized channel output be denoted by $\widehat{W}$.
Since the conditional measure $P_{\widehat{W}|X}(\cdot|x), x\in\mathbb{R}$ is continuous in $x$, and $\lim_{x\to \infty} P_{\widehat{W}|X}(\mathcal{A}|x) = \mathbbm{1}(\hat{w}\in \mathcal{A}), \mathcal{A}\in \mathbb{B}$ for some fixed $\hat{w}$, the conditions in the proof of \cite[Prop. 1]{singh2009limits} hold, and the optimal input distribution has bounded support. From the extension of Witsenhausen's result \cite{witsenhausen1980some} given in \cite[Prop. 2]{singh2009limits}, the optimal input distribution is discrete and takes at most ${2n_q+1}$ values. This completes the proof of converse. To prove achievability, it suffices to show that one can choose the set of functions $f_{j}(\cdot), j\in [n_q]$ and quantization thresholds $t_{j,1}, j\in [n_q]$ such that the resulting quantizer operates as described in the theorem statement, i.e., it generates $\widehat{W}=\mathbf{q}(hX+N)$ where $\mathbf{q}(y)=k$ if $y\in [t_{k},t_{k+1}], k\in \{1,\cdots,2n_q]$ and $\mathbf{q}(y)=0$ if $y>t_{2n_q}$ or $y<t_{1}$. To this end, let $\mathbf{t}^*$ be the optimal quantizer thresholds in \eqref{eq:th:1}. Let $r_1,r_2,\cdots,r_{2n_q}$ be the elements of $\mathbf{t}^*$ written in non-decreasing order. Define a quantizer with associated analog functions $f_{j,1}(y)\triangleq  |y-\frac{r_j+r_{n_q+j}}{2}|$ and $t_{j,1}=\frac{r_{n_q+j}-r_j}{2}$. Note that $t_j>0$ since $r_j, j\in [n_q]$ are non-decreasing. Then, similar to the proof of Proposition \ref{Prop:1}, the quantization rule gives distinct outputs for $y\in [r_{k},r_{k+1}], k\in \{1,\cdots,2n_q]$ and $y\in [r_{2n_q},\infty) \cup [-\infty,r_{1}]
$ as desired.
\qed

\section{Proof of Proposition \ref{Prop:3}}
\label{App:Prop:5}
For $j\in [n_q]$ and non-decreasing vector $(r_{i_1},r_{i_2},\cdots, r_{i_{2^\delta}})$ where $i_\lambda\in \mathcal{I}_j, j\in [n_q], \lambda\in [2^{\delta}]$, define 
\begin{align*}
&a_{1,j}\triangleq  \frac{r_{i_1}+r_{i_{2^{\delta}}}}{2}, \qquad a_{s,j}\triangleq  \frac{r_{i_{2^{\delta}}}+r_{i_{\eta_s}}}{2}-\sum_{s'=1}^{s-1} a_{s',j}, 1<s\leq \delta,\\
 &t_{j}\triangleq r_{i_{2^{\delta}}}-\sum_{s'=1}^{\delta-1} a_{s',j}, 1<s\leq \delta
\end{align*}
where $\eta_s\triangleq 2^{\delta}- \sum_{s'=1}^{s-1}2^{\delta-s'}+1, s>1$. 
Consider a  quantizer $\mathbf{q}(\cdot)$ with ADC thresholds  $t(1:n_q)$ and associated analog functions $f_{j}(y)\triangleq A_{\delta}(x,a^{\delta}), j\in [n_q]$. Then, $r_1,r_2,\cdots, r_{\gamma}$ are the non-decreasing sequence of roots of $f_{j}(\cdot), j\in [n_q]$, and the associated code of the quantizer $\mathbf{q}(\cdot)$ is $\mathcal{C}$ as desired.
\qed

\section{Proof of Proposition \ref{Prop:6}}
\label{App:Prop:6}
\label{App:Prop:4}
We provide an outline of the proof.  Let us consider the following cases:
\\\textbf{Case 1:} $\sum_{j=1}^{n_q }\kappa_j\geq 2^{n_q}$
\\In this case, one can use a balanced Gray code \cite{bhat1996balanced} to construct $\mathcal{C}$. A balanced Gray code is a (binary) code where consecutive codewords have Hamming distance equal to one, and each of the bit positions changes value either $2\floor{\frac{2^{n_q}}{2n_q}}$ times or $2\ceil{\frac{2^{n_q}}{2n_q}}$ times. If $\min_{j\in [n_q]} \kappa_j\geq 2\ceil{\frac{2^{n_q}}{2n_q}}$ the proof is complete as one can concatenate the balanced gray code with a series of additional repeated codewords to satisfy the transition counts, and since the balanced gray code is a subcode of the resulting code, we have $|\mathcal{C}|=2^{n_q}$. Otherwise,  there exists $j\in [n_q]$ such that $\kappa_j< 2\ceil{\frac{2^{n_q}}{2n_q}}$.
In this case, without loss of generality, let us assume that $\kappa_1\leq \kappa_2,\cdots \leq \kappa_{n_q}$. Note that since $|\kappa_j-\kappa_j'|\leq 2, j,j'\in [n_q]$ and $\kappa_j, j\in [n_q]$ are even, there is at most one $j^*\in [n_q]$ such that $\kappa_{j^*}\leq \kappa_{j^*+1}$. Let $\kappa'_1,\kappa'_2,\cdots,\kappa'_{n_q}$ be the transition count sequence of a balanced gray code $\mathcal{C}'$ written in non-decreasing order. 
Note that $2\ceil{\frac{2^{n_q}}{2n_q}}-2\floor{\frac{2^{n_q}}{2n_q}}=2$. Hence, similar to the above argument, there can only be one $j'\in [n_q]$ for  which   $\kappa_{j'}\leq \kappa_{j'+1}$. 
Since $\sum_{j=1}^{n_q} \kappa_j\geq  2^{n_q}= \sum_{j=1}^{n_q} \kappa'_j$, we must have $j^*\leq j'$. So, the balanced gray code can be used as a subcode similar to the previous case by correctly ordering the bit positions to match the order of $\kappa_j, j\in [n_q]$. This completes the proof. 
\\\textbf{Case 2:} $\sum_{j=1}^{n_q}\kappa_j< 2^{n_q}$
\\The proof is based on techniques used in the construciton of balanced Gray codes \cite{bhat1996balanced}.
We prove the result by induction on $n_q$.
The proof for $n_q=1,2$ is straightforward and follows by construction of length-one and length-two sequences. For $n_q>2$, Assume that the result holds for all $n'_q\leq n_q$. Without loss of generality, assume that $\kappa_1\leq \kappa_2,\leq \cdots \leq \kappa_{n_q}$. The proof considers four sub-cases as follows.
\\\textbf{Case 2.i:} $\sum_{j=3}^{n_q}\kappa_j\in [0,2^{n_q-2}]$
\\In this case, by the induction assumption, there exists $\mathcal{C}'$,  a code with codewords of length $n_q-2$, whose transition sequence is $\kappa_3,\kappa_4,\cdots,\kappa_{n_q}$, and $|\mathcal{C}'|= \sum_{j=3}^{n_q}\kappa_j$. We construct $\mathcal{C}$ from $\mathcal{C}'$ as follows. Let $\mathbf{c}_{0}=(0,0,\mathbf{c}'_0)$, $\mathbf{c}_{1}=(0,1,\mathbf{c}'_0)$, $\mathbf{c}_{2}=(1,1,\mathbf{c}'_0)$,
$\mathbf{c}_{3}=(1,0,\mathbf{c}'_0)$,
$\mathbf{c}_{4}=(1,0,\mathbf{c}'_1)$,
$\mathbf{c}_{5}=(0,0,\mathbf{c}'_1)$, $\mathbf{c}_{6}=(0,1,\mathbf{c}'_1)$,
$\mathbf{c}_{7}=(1,1,\mathbf{c}'_1)$,$\cdots$. This resembles the procedure for constructing balanced gray codes \cite{bhat1996balanced}. We continue concatenating the first two bits of each codeword in $\mathcal{C}$ to the codewords in $\mathcal{C}'$ using the procedure described above until $\kappa_1$ transitions for position 1 and $\kappa_2$ transitions for position 2 have taken place. Note that this is always possible since i) for each two codewords in $\mathcal{C}'$, we `spend' two transitions of each of the first and second positions in $\mathcal{C}$ to produce four new codewords, ii) $\kappa_2-\kappa_1\leq 2$, and iii) $\kappa_2\leq \sum_{j=3}^{n_q}\kappa_j$, where the latter condition ensures that we do not run out of codewords in $\mathcal{C}'$ before the necessary transitions in positions 1 and 2 are completed. After $\kappa_2+1$ codewords, the transitions in positions 1 and 2 are completed, and the last produced codeword is $(0,0,\mathbf{c}'_{\kappa_2+1})$  since $\kappa_1$ and $\kappa_2$ are both even. To complete the code $\mathcal{C}$, we add $(0,0,\mathbf{c}'_{i}), i\in [\kappa_2+2,\sum_{j=3}^{n_q}\kappa_j]$. Then, by construction, we have $|\mathcal{C}|=|\mathcal{C}'|+\kappa_1+\kappa_2=\sum_{j=1}^{n_q}\kappa_j$ and the code satisfied Properties 1), 2), 3), and 5) in Proposition \ref{Prop:2}. 
\\\textbf{Case 2.ii:}$\sum_{j=3}^{n_q}\kappa_j\in [2^{n_q-2}, 2^{n_q-1}]$
\\ Similar to the previous case, let $\mathcal{C}'$ be a balanced gray code with codeword length $n_q-2$ and transition counts $\kappa'_1\leq \kappa'_2\leq \cdots\leq \kappa'_{n_q-2}$. Define $\kappa''_{j}=\kappa_j- \kappa'_{j+2}, j\in \{3,4,\cdots,n_q\}$. Note that $\kappa''_j$ satisfy the conditions on transition counts in the proposition statement, and hence by the induction assumption, there exists a code $\mathcal{C}''$ with transition counts $\kappa''_j, j\in [n_q-2]$. The proof is completed by appropriately concatenating $\mathcal{C}'$ and $\mathcal{C}''$ to construct $\mathcal{C}$. Let $\gamma''$ be the number of codewords in $\mathcal{C}''$ and define $\mathbf{c}_i=(0,0,\mathbf{c}''_i), i\in [\gamma'']$, $\mathbf{c}_{\gamma''+1}=(0,1,\mathbf{c}''_{\gamma''})$,  $\mathbf{c}_{\gamma''+2}=(1,1,\mathbf{c}''_{\gamma''})$, $\mathbf{c}_{\gamma''+3}=(1,0,\mathbf{c}''_{\gamma''})$, $\mathbf{c}_{\gamma''+4}=(1,0,\mathbf{c}'_{1})$,$\cdots$. Similar to the previous case, it is straightforward to show that this procedure yields a code $\mathcal{C}$ with the desired transition sequence. 

The proof for the two subcases where $\sum_{j=3}^{n_q}\kappa_j\in [2^{n_q-1},3 \times 2^{n_q-2}]$ and $\sum_{j=3}^{n_q}\kappa_j\in [3\times 2^{n_q-1}, \times 2^{n_q-1}]$ is similar and is omitted for brevity.

\end{appendices}

\end{document}